\documentclass[sigplan,screen,authorversion]{acmart}

\setcopyright{acmlicensed}
\acmPrice{15.00}
\acmDOI{10.1145/3519939.3523445}
\acmYear{2022}
\copyrightyear{2022}
\acmSubmissionID{pldi22main-p150-p}
\acmISBN{978-1-4503-9265-5/22/06}
\acmConference[PLDI '22]{Proceedings of the 43rd ACM SIGPLAN International Conference on Programming Language Design and Implementation}{June 13--17, 2022}{San Diego, CA, USA}
\acmBooktitle{Proceedings of the 43rd ACM SIGPLAN International Conference on Programming Language Design and Implementation (PLDI '22), June 13--17, 2022, San Diego, CA, USA}

\usepackage{booktabs}   
\usepackage{subcaption} 

\usepackage{stmaryrd}
\usepackage[shortlabels]{enumitem}
\usepackage{algorithmic}
\usepackage{cprotect}
\usepackage{listings, listings-rust}
\usepackage{mathpartir}
\usepackage{xfrac}
\usepackage{soul}
\usepackage{calc}
\usepackage{fancyvrb}
\usepackage{subcaption}

\usepackage{bm}

\usepackage{cleveref}
\usepackage[inference]{semantic}
\usepackage[verbose]{wrapfig}
\usepackage{xspace}

\newcommand{\mathdef}[1]{\hypertarget{def:#1}{}}

\newenvironment{blacklinks}
  {\hypersetup{linkcolor=.}}
  {\hypersetup{linkcolor=ACMPurple}}
  
\newcommand{\blackref}[2]{\begin{blacklinks}\hyperlink{#1}{#2}\end{blacklinks}}  

\newcommand{\updateconflicts}[3]{\blackref{def:updateconflicts}{\msf{update}\text{-}\msf{conflicts}}\left(#1,#2,#3\right)}

\newcommand{\refs}[2]{\blackref{def:refs}{{#1}\text{-}\msf{refs}}(#2)}

\newcommand{\loans}[2]{\hyperlink{def:loans}{{#1}\text{-}\msf{loans}}(#2)}

\theoremstyle{definition}

\Crefname{principle}{Principle}{Principles}

\makeatletter
\xdef\@thefnmark{\@empty}

\makeatother
\newcounter{typerule}
\crefformat{typerule}{#2\scshape{#1}#3}
\Crefformat{typerule}{#2\scshape{#1}#3}

\newcommand{\typeruleInt}[4]{
	\def\thetyperule{#1}%
	\refstepcounter{typerule}%
	\label{tr:#4}%
  \inferrule[#1]{#2}{#3}
}
\newcommand{\typerule}[4]{
  \typeruleInt{#1}{#2}{#3}{#4}
}

\makeatletter
\newenvironment{btHighlight}[1][]
{\begingroup\tikzset{bt@Highlight@par/.style={#1}}\begin{lrbox}{\@tempboxa}}
{\end{lrbox}\bt@HL@box[bt@Highlight@par]{\@tempboxa}\endgroup}

\newcommand\btHL[1][]{%
  \begin{btHighlight}[#1]\bgroup\aftergroup\bt@HL@endenv%
}
\newcommand{\btHLval}{\btHL[fill=ForestGreen!30]}

\def\bt@HL@endenv{%
  \end{btHighlight}%
  \egroup
}
\newcommand{\bt@HL@box}[2][]{%
  \tikz[#1]{%
    \pgfpathrectangle{\pgfpoint{1pt}{0pt}}{\pgfpoint{\wd #2}{\ht #2}}%
    \pgfusepath{use as bounding box}%
    \node[anchor=base west, fill=Peach!30,outer sep=0pt,inner xsep=1pt, inner ysep=0pt, rounded corners=2pt, minimum height=\ht\strutbox+1pt,#1]{\raisebox{1pt}{\strut}\strut\usebox{#2}}; 
  }%
}
\makeatother

\definecolor{pred}{RGB}{230,159,0}
\definecolor{pgrn}{RGB}{86,180,233}
\definecolor{pblu}{RGB}{0,158,115}
\definecolor{ppnk}{RGB}{240,228,66}
\definecolor{ppur}{RGB}{213,94,0}
\definecolor{porg}{RGB}{204,121,167}

\definecolor{Peach}{RGB}{253,137,85}
\definecolor{ForestGreen}{RGB}{0,144,79}

\setul{3pt}{1pt}

\fvset{fontsize=\small}

\newcommand{\msf}[1]{\mathsf{#1}}
\newcommand{\msfb}[1]{\text{\textbf{\small\texttt{#1}}}}
\newcommand{\rust}{\lstinline}

\newcommand{\shrd}{\msf{shrd}}
\newcommand{\uniq}{\msf{uniq}}

\newcommand{\bnfm}[3]{\msf{#1}~{#2} ::=&\ {#3}}
\newcommand{\tyvar}[1]{\tau^\textsc{#1}}
\newcommand{\tysize}{\tyvar{si}}

\newcommand{\eref}[3]{\&{#2}\,{#1}\,{#3}}

\newcommand{\fndef}{
    \msfb{fn}~f \langle \overline{\varphi}, \overline{\varrho}, \overline{\alpha}, \overline{\varrho_1 : \varrho_2} \rangle (x : \tysize_a) \rightarrow \tysize_r  ~ \{ ~ e ~ \}
}
\newcommand{\uty}{\msf{u32}}
\newcommand{\unitty}{\msf{unit}}

\newcommand{\olet}[4]{\msfb{let}~{#1} : {#2}~=~{#3};~{#4}}   
\newcommand{\letprov}[2]{\msfb{letprov}\langle {#1} \rangle ~ {#2}}
\newcommand{\framed}[1]{\msf{framed}~{#1}}

\newcommand{\disjoint}[2]{{#1} \blackref{def:disjoint}{\mathrel{\#}} {#2}}
\newcommand{\notdisjoint}[2]{{#1} \blackref{def:notdisjoint}{\mathop{\sqcap}} {#2}}
\newcommand{\loanset}{\{\overline{l}\}}
\newcommand{\evalsto}[4]{\Sigma \vdash ({#1}; {#2}) \xrightarrow{\ast} ({#3}; {#4})}
\newcommand{\stepsto}[4]{\Sigma \vdash ({#1}; {#2}) \rightarrow ({#3}; {#4})}  
\newcommand{\valuectx}{\mathcal{V}}
\newcommand{\stacksep}{\mathrel{\natural}}
\newcommand{\tc}[6]{{#1}; {#2}; {#3} \vdash {#4} : {#5} \Rightarrow {#6}}
\newcommand{\tcnew}[8]{{#1}; {#2}; {#3}; {#4} \vdash {#5} : {#6} \Rightarrow {#7}; {#8}}
\definecolor{highlight}{RGB}{180,30,30}
\newcommand{\hlight}[1]{\textcolor{highlight}{#1}}
\newcommand{\hlt}[1]{\hlight{#1}}
\definecolor{deemphasize}{RGB}{120,120,120}
\newcommand{\gray}[1]{\textcolor{deemphasize}{#1}}
\newcommand{\loc}{\ell}
\newcommand{\withslice}[1]{~\bullet~{#1}}
\newcommand{\medcup}{\textstyle \bigcup}
\newcommand{\stepped}[1]{#1'}
\newcommand{\deps}[2]{\blackref{def:deps}{\msf{deps}}({#1}, {#2})}

\newcommand{\baseline}{\textsc{Modular}}
\newcommand{\wholeprogram}{\textsc{Whole-program}}
\newcommand{\mutblind}{\textsc{Mut-blind}}
\newcommand{\pointerblind}{\textsc{Ref-blind}}

\newcommand{\eqdef}{%
    ~\mathop{\overset{\mathrm{def}}{\resizebox%
      {\widthof{\ensuremath{\mathop{\overset{\mathrm{def}}{=}}}}}%
      {\heightof{=}}{=}}}~}

\newcommand{\savetheorem}[1]{\newtheorem*{#1}{Theorem~\ref{thm:#1}}}

\newcommand{\stackeq}[1]{\blackref{def:stackeq}{\mathop{\sim_{#1}}}}

\newcommand{\stackdepeq}[2]{\blackref{def:stackdepeq}{\mathop{\sim^{#1}_{#2}}}}

\newcommand{\mut}{\msf{mut}}
\newcommand{\any}{\msf{any}} 
\newcommand{\loan}{\msf{loan}}

\newcommand{\arrg}{\msf{arg}}
\newcommand{\rd}{\msf{read}}

\newcommand{\mutptr}{\msf{ptr}}

\newenvironment{proofcases}[1][*]
  {\begin{enumerate}[(a),itemsep=2pt,leftmargin=#1]}
  {\end{enumerate}}
\newenvironment{proofsteps}[1][*]
  {\begin{enumerate}[1.,itemsep=2pt,leftmargin=#1]}
  {\end{enumerate}}

\newcommand{\sysname}{Flowistry}

\begin{document}

\title{Modular Information Flow through Ownership}

\author{Will Crichton}
\orcid{0000-0001-8639-6541}
\affiliation{%
  \institution{Stanford University}
  \city{Stanford}
  \country{USA}
}
\email{wcrichto@cs.stanford.edu}

\author{Marco Patrignani}
\orcid{0000-0003-3411-9678}    
\affiliation{
  \institution{University of Trento}
  \city{Trento}
  \country{Italy}
}
\email{marco.patrignani @ unitn.it}         

\author{Maneesh Agrawala}
\author{Pat Hanrahan}
\affiliation{%
  \institution{Stanford University}
  \city{Stanford}
  \country{USA}
}

\begin{abstract}
Statically analyzing information flow, or how data influences other data within a program, is a challenging task in imperative languages. Analyzing pointers and mutations requires access to a program's complete source. However, programs often use pre-compiled dependencies where only type signatures are available. We demonstrate that ownership types can be used to soundly and precisely analyze information flow through function calls given only their type signature. From this insight, we built Flowistry, a system for analyzing information flow in Rust, an ownership-based language. We prove the system's soundness as a form of noninterference using the Oxide formal model of Rust. Then we empirically evaluate the precision of Flowistry, showing that modular flows are identical to whole-program flows in 94\% of cases drawn from large Rust codebases. We illustrate the applicability of Flowistry by using it to implement prototypes of a program slicer and an information flow control system.
\end{abstract}



\maketitle

\section{Introduction}
\label{sec:intro}

Information flow describes how data influences other data within a program. Information flow has applications to security, such as information flow control\,\cite{sabelfeld2003language}, and to developer tools, such as program slicing\,\cite{weiser1984program}. Our goal is to build a practical system for analyzing information flow, meaning:

\begin{enumerate}[leftmargin=*]
    \item \textbf{Applicable to common language features:} the language being analyzed  should support widely used features like pointers and in-place mutation. 
    \item \textbf{Zero configuration to run on existing code:} the analyzer must integrate with an existing language and existing unannotated programs. It must not require users to adopt a new language designed for information flow.
    \item \textbf{No dynamic analysis:} to reduce integration challenges and costs, the analyzer must be purely static --- no modifications to runtimes or binaries are needed.
    \item \textbf{Modular over dependencies:} programs may not have source available for dependencies. The analyzer must have reasonable precision without whole-program analysis.
\end{enumerate}

As a case study on the challenges imposed by these requirements, consider analyzing the information that flows to the return value in this C++ function:
\smallskip
\begin{lstlisting}[language=C++]
// Copy elements 0 to max into a new vector
vector<int> copy_to(vector<int>& v, size_t max) {
  vector<int> v2; size_t i = 0;
  for (auto x(v.begin()); x != v.end(); ++x) {
    if (i == max) { break; }
    v2.push_back(*x); ++i;
  }
  return v2;
}
\end{lstlisting}
\smallskip
Here, a key flow is that \verb|v2| is influenced by \verb|v|: (1) \verb|push_back| mutates \verb|v2| with \verb|*x| as input, and (2) \verb|x| points to data within \verb|v|. But how could an analyzer statically deduce these facts? For C++, the answer is \textit{by looking at function implementations}. The implementation of \verb|push_back| mutates \verb|v2|, and the implementation of \verb|begin| returns a pointer to data in \verb|v|. 

However, analyzing such implementations violates our fourth requirement, since these functions may only have their type signature available. In C++, given only a function's type signature, not much can be inferred about its behavior, since the type system does not contain information relevant to pointer analysis.

Our key insight is that \textit{ownership types} can be leveraged to modularly analyze pointers and mutation using only a function's type signature. Ownership has emerged from several intersecting lines of research on linear logic\,\cite{girard1987linear}, class-based alias management\,\cite{clarke1998ownership}, and region-based memory management\,\cite{grossman2002region}. The fundamental law of ownership is that data cannot be simultaneously aliased and mutated. Ownership-based type systems enforce this law by tracking which entities own which data, allowing ownership to be transferred between entities, and flagging ownership violations like mutating immutably-borrowed data.

Today, the most popular ownership-based language is Rust. Consider the information flows in this Rust implementation of \verb|copy_to|:
\smallskip
\begin{lstlisting}
fn copy_to(v: &Vec<i32>, max: usize) -> Vec<i32> {
  let mut v2 = Vec::new();
  for (i, x) in v.iter().enumerate() {
    if i == max { break; }
    v2.push(*x);
  }
  return v2;
}
\end{lstlisting}
\smallskip
\noindent Focus on the two methods \lstinline|push| and \lstinline|iter|. For a \lstinline|Vec<i32>|, these methods have the following type signatures:
\begin{lstlisting}
fn push(&mut self, value: i32);
fn iter<'a>(&'a self) -> Iter<'a, i32>;
\end{lstlisting}
To determine that \lstinline|push| mutates \lstinline|v2|, we leverage \textit{mutability modifiers}. All references in Rust are either immutable (i.e. the type is \lstinline|&T|) or mutable (the type is \lstinline|&mut T|). Therefore \lstinline|iter| does not mutate \lstinline|v| because it takes \lstinline|&self| as input (excepting interior mutability, discussed in \Cref{sec:limitations}), while \lstinline|push| may mutate \lstinline|v2| because it takes \lstinline|&mut self| as input.
 
To determine that \lstinline|x| points to \lstinline|v|, we leverage \textit{lifetimes}. All references in Rust are annotated with a lifetime, either explicitly (such as \lstinline|'a|) or implicitly. Shared lifetimes indicate aliasing: because \lstinline|&self| in \lstinline|iter| has lifetime \lstinline|'a|, and because the returned \verb|Iter| structure shares that lifetime, then we can determine that \lstinline|Iter| may contain pointers to \lstinline|self|.

Inspired by this insight, we built \sysname{}, a system for analyzing information flow in the safe subset of Rust programs. \sysname{} satisfies our four design criteria: (1) Rust supports pointers and mutation, (2) \sysname{} does not require any change to the Rust language or to Rust programs, (3) \sysname{} is a purely static analysis, and (4) Flowistry uses ownership types to analyze function calls without needing their definition. This paper presents a theoretical and empirical investigation into \sysname{} in five parts:
\begin{enumerate}[leftmargin=*]
    \item We provide a precise description of how \sysname{} computes information flow by embedding its definition within Oxide\,\cite{weiss2019oxide}, a formal model of Rust (\Cref{sec:analysis}).
    \item We prove the soundness of our information flow analysis as a form of noninterference (\Cref{sec:soundness}).
    \item We describe the implementation of \sysname{} that bridges the theory of Oxide to the practicalities of Rust (\Cref{sec:implementation}).
    \item We evaluate the precision of the modular analysis on a dataset of large Rust codebases, finding that modular flows are identical to whole-program flows in 94\% of cases, and are on average 7\% larger in the remaining cases (\Cref{sec:evaluation}).
    \item We demonstrate the utility of \sysname{} by using it to prototype a program slicer and an IFC checker (\Cref{sec:applications}).
\end{enumerate}

We conclude by presenting related work (\Cref{sec:rw}) and discussing future directions for \sysname{} (\Cref{sec:discussion}).
Due to space constraints, we omit many formal details, all auxiliary lemmas, and all proofs.
The interested reader can find them in \Cref{sec:appendix}.
\sysname{} and our applications of it are publicly available, open-source, MIT-licensed projects at \url{https://github.com/willcrichton/flowistry}.

\section{Analysis}
\label{sec:analysis}

Inspired by the dependency calculus of Abadi et al.\,\cite{abadi1999core}, our analysis represents information flow as a set of dependencies for each variable in a given function. The analysis is flow-sensitive, computing a different dependency set at each program location, and field-sensitive, distinguishing between dependencies for fields of a data structure.

While the analysis is implemented in and for Rust, our goal here is to provide a description of it that is both concise (for clarity of communication) and precise (for amenability to proof). We therefore base our description on Oxide\,\cite{weiss2019oxide}, a formal model of Rust. At a high level, Oxide provides three ingredients: 
\begin{enumerate}[leftmargin=*]
    \item A syntax of Rust-like programs with expressions $e$ and types $\tau$.
    \item A type-checker, expressed with the judgment $\tc{\Sigma}{\Delta}{\Gamma}{e}{\tau}{\Gamma'}$ using the contexts $\Gamma$ for types and lifetimes, $\Delta$ for type variables, and $\Sigma$ for global functions.
    \item An interpreter, expressed by a small-step operational semantics with the judgment $\stepsto{\sigma}{e}{\stepped{\sigma}}{e'}$ using $\sigma$ for a runtime stack.
\end{enumerate}

We extend this model by assuming that each expression in a program is automatically labeled with a unique location $\loc$. Then for a given expression $e$, our analysis computes the set of dependencies $\kappa ::= \{\overline{\loc}\}$. Because expressions have effects on persistent memory, we further compute a \textit{dependency context} $\Theta ::= \{\overline{p \mapsto \kappa}\}$ from memory locations $p$ to dependencies $\kappa$. The computation of information flow is intertwined with type-checking, represented as a modified type-checking judgment (additions highlighted in red):
$$
\tcnew{\Sigma}{\Delta}{\Gamma}
  {\hlt{\Theta}}
  {e_{\hlt{\loc}}}
  {\tau\hlt{\withslice{\kappa}}}
  {\Gamma'}
  {\hlt{\Theta'}}
$$

\noindent This judgment is read as, ``with type contexts $\Sigma, \Delta, \Gamma$ and dependency context $\hlt{\Theta}$, $e$ at location $\hlt{\loc}$ has type $\tau$ and dependencies $\hlt{\kappa}$, producing a new dependency context $\hlt{\Theta'}$.''

Oxide is a large language --- describing every feature, judgment, and inference rule would exceed our space constraints. Instead, in this section we focus on a few key rules that demonstrate the novel aspects of our system. We first lay the foundations for dealing with variables and mutation (\Cref{sec:places}), and then describe how we modularly analyze references (\Cref{sec:references}) and function calls (\Cref{sec:funcalls}). 
The remaining rules can be found in \Cref{sec:more_rules}.


\subsection{Variables and mutation}
\label{sec:places}

The core of Oxide is an imperative calculus with constants and variables. The abstract syntax for these features is below:
\begin{gather*}
\begin{aligned}
\msf{Variable}~x \hspace{12pt} \msf{Number}~n 
\end{aligned}
\\
\begin{aligned}
\bnfm{Path}{q}{\varepsilon \mid n.q} \\
\bnfm{Place}{\pi}{x.q} \\
\bnfm{Constant}{c}{() \mid n \mid \msf{true} \mid \msf{false}} \\
\bnfm{Base~Type}{\tyvar{b}}{\unitty\mid \uty \mid \msf{bool}}  \\
\bnfm{Sized~Type}
  {\tyvar{si}}
  {\tyvar{b} \mid 
   (\tyvar{si}_1, \ldots, \tyvar{si}_n)} \mid \ldots \\
\bnfm{Expression}{e}{}
    c \mid \pi \mid \olet{x}{\tyvar{si}_a}{e_1}{e_2} \mid \\
    &\pi := e \mid e_1;~e_2 \mid \ldots
\end{aligned}
\end{gather*}

Constants are Oxide's atomic values and also the base-case for information flow. A constant's dependency is simply itself, expressed through the \Cref{tr:tu32} rule:
\vspace{4pt}
\begin{mathpar}
\typerule{T-u32}
  {\ }
  {
    \tcnew{\Sigma}{\Delta}{\Gamma}
    {\hlt{\Theta}}
    {n_{\hlt{\loc}}}
    {\uty \hlt{\withslice{\{\loc\}}}}
    {\Gamma}{\hlt{\Theta}}
  }
  {tu32}
\end{mathpar}

Variables and mutation are introduced through let-bindings and assignment expressions, respectively. For example, this (location-annotated) program mutates a field of a tuple:
\[
    \olet{t}{(\uty, \uty)}{(1_{\loc_1}, 2_{\loc_2})}{t.1 := 3_{\loc_3}}
\]

Here, $t$ is a variable and $t.1$ is a \textit{place}, or a description of a specific region in memory. For information flow, the key idea is that let-bindings introduce a set of places into $\Theta$, and then assignment expressions change a place's dependencies within $\Theta$.
In the above example, after binding $t$, then $\Theta$ is: 
$$\Theta = \{t, t.0, t.1 \mapsto \{\loc_1, \loc_2\}\}$$
After checking ``$t.1 := 3$'', then $\loc_3$ is added to $\Theta(t)$ and $\Theta(t.1)$, but not $\Theta(t.0)$. This is because the values of $t$ and $t.1$ have changed, but the value of $t.0$ has not. Formally, the let-binding rule is:
\vspace{4pt}
\begin{mathpar}
\typerule{T-Let}
  {
    \tcnew{\Sigma}{\Delta}{\Gamma}{\hlight{\Theta}}{e_1}{\tyvar{si}_1 \hlight{\withslice{\kappa_1}}}{\Gamma_1}{\hlight{\Theta_1}} 
    \\
    \gray{\Gamma; \Delta_1 \vdash \tyvar{si}_1 \lesssim \tyvar{si}_a \Rightarrow \Gamma'_1}
    \\
    \hlight{\Theta_1' = \Theta_1[\forall \pi = \pi^\square[x] ~ . ~ \pi \mapsto \kappa_1]}
    \\
    \tcnew{\Sigma}{\Delta}
      {\text{gc-loans}(\Gamma'_1, x : \tyvar{si}_a)}
      {\hlight{\Theta_1'}}{e_2}
      {\tyvar{si}_2 \hlight{\withslice{\kappa_2}}}
      {\Gamma_2, x : \tyvar{sd}}{\hlight{\Theta_2}}
  }
  {
    \tcnew{\Sigma}{\Delta}{\Gamma}{\hlight{\Theta}}
    {\olet{x}{\tyvar{si}_a}{e_1}{e_2}}
    {\tyvar{si}_2 \hlight{\withslice{\kappa_2}}}
    {\Gamma_2}{\hlight{\Theta_2}}
  }{tlet}
\end{mathpar}

Again, this rule (and many others) contain aspects of Oxide that are not essential for understanding information flow such as the subtyping judgment $\tau_1 \lesssim \tau_2$ or the metafunction $\mathsf{gc}\text{-}\mathsf{loans}$. For brevity we will not cover these aspects here, and instead refer the interested reader to \citet{weiss2019oxide}. 
We have deemphasized (in grey) the judgments which are not important to understanding our information flow additions.

The key concept is the formula $\Theta_1[\forall \pi = \pi^\square[x] ~ . ~ \pi \mapsto \kappa_1]$. This introduces two shorthands: first, $\pi^\square[x]$ means ``a place $\pi$ with root variable $x$ in a context $\pi^\square$'', used to decompose a place. In \Cref{tr:tlet}, the update to $\Theta_1$ happens for all places with a root variable $x$. Second, $\Theta_1[\pi \mapsto \kappa_1]$ means ``set $\pi$ to $\kappa_1$ in $\Theta_1$''. So this rule specifies that when checking $e_2$, all places within $x$ are initialized to the dependencies $\kappa_1$ of $e_1$.

Next, the assignment expression rule is defined as updating all the \textit{conflicts} of a place $\pi$:
\vspace{4pt}
\begin{mathpar}
\typerule{T-Assign}
  {
    \tcnew{\Sigma}{\Delta}{\Gamma}{\hlight{\Theta}}{e}{\tyvar{si} \hlight{\withslice{\kappa}}}{\Gamma_1}{\hlight{\Theta_1}} 
    \\
    \gray{\Gamma_1(\pi) = \tyvar{sx} }
    \\
    \gray{(\tyvar{sx} = \tyvar{sd} \vee \Delta; \Gamma_1 \vdash_\uniq{} \pi \Rightarrow \{~^\uniq{} \pi\}) }
    \\
    \gray{\Delta; \Gamma_1 \vdash \tyvar{si} \lesssim \tyvar{sx} \Rightarrow \Gamma'}
    \\
    \hlight{\Theta_2 = \Theta_1[\updateconflicts{\Theta_1}{\pi}{\kappa}]}
  }
  {
    \tcnew{\Sigma}{\Delta}{\Gamma}{\hlight{\Theta}}{\pi := e}{\unitty \hlight{\withslice{\varnothing}}}
   {\Gamma'[\pi \mapsto \tyvar{si}] \vartriangleright \pi}{\hlight{\Theta_2}}
  }{tassign}
\end{mathpar}

If you conceptualize a type as a tree and a path as a node in that tree, then a node's conflicts are its ancestors and descendants (but not siblings). Semantically, conflicts are the set of places whose value change if a given place is mutated. Recall from the previous example that $t.1$ conflicts with $t$ and $t.1$, but not $t.0$. Formally, we say two places are disjoint ($\#$) or conflict ($\sqcap$) when:
\mathdef{disjoint}
\mathdef{notdisjoint}
\begin{align*}
    \disjoint{x_1.q_1}{x_2.q_2} &\eqdef x_1 \neq x_2 \vee (&&(q_1 \text{ is not a prefix of } q_2) ~ \wedge \\ & &&(q_2 \text{ is not a prefix of } q_1)) \\
    \notdisjoint{\pi_1}{\pi_2} &\eqdef \neg(\disjoint{\pi_1}{\pi_2}) 
\end{align*}

Then to update a place's conflicts in $\Theta$, we define the metafunction $\msf{update}\text{-}\msf{conflicts}$ to add $\kappa$ to all conflicting places $p'$. (Note that this rule is actually defined over place \textit{expressions} $p$, which are explained in the next subsection.)
\mathdef{updateconflicts}%
\begin{align*}
    &\updateconflicts{\Theta}{p}{\kappa} \eqdef \\ 
    &\hspace{20pt} \forall p' \mapsto \kappa_{p'} \in \Theta_\msf{cfl}  ~ . ~ p' \mapsto \kappa_{p'} \cup \kappa \\
    &\hspace{20pt} \text{where} ~ \Theta_\msf{cfl} = \{p' \mapsto \kappa_{p'} \in \Theta \mid \notdisjoint{p}{p'}\}
\end{align*}

Finally, the rule for reading places is simply to look up the place's dependencies in $\Theta$:
\vspace{4pt}
\begin{mathpar}
\typerule{T-Move}{
    \gray{\Delta; \Gamma \vdash_\uniq{} \pi \Rightarrow \{~^\uniq{} \pi\} }
    \qquad
    \gray{\Gamma(\pi) = \tyvar{si} }
    \quad
    \gray{\msf{noncopyable}_\Sigma~\tyvar{si} }
  }{
    \tcnew{\Sigma}{\Delta}{\Gamma}{\hlight{\Theta}}{\pi}
    {\tyvar{si} \hlight{\withslice{\Theta(\pi)}}}
    {\Gamma[\pi \mapsto \tau^{\textsc{si}^\dagger}]}
    {\hlight{\Theta}}
  }{tmove}
\end{mathpar}

\subsection{References}
\label{sec:references}

Beyond concrete places in memory, Oxide also contains references that point to places. As in Rust, these references have both a lifetime (called a ``provenance'') and a mutability qualifier (called an ``ownership qualifier''). Their syntax is:
\begin{gather*}
\begin{aligned}
\msf{Concrete~Provenance}~r \hspace{12pt} \msf{Abstract~Provenance}~\varrho 
\end{aligned}
\\
\begin{aligned}
\bnfm{Place~Expression}{p}{x \mid \ast p \mid p.n} \\
\bnfm{Provenance}{\rho}{\varrho \mid r} \\
\bnfm{Ownership~Qualifier}{\omega}{\shrd \mid \uniq{}} \\
\bnfm{Sized~Type}{\tyvar{si}}{\ldots \mid \eref{\omega}{\rho}{\tyvar{xi}}} \\ 
\bnfm{Expression}{e}{\ldots \mid \eref{\omega}{r}{p} \mid p := e \mid \letprov{r}{e}}
\end{aligned}
\end{gather*}

Provenances are created via a $\msfb{letprov}$ expression, and references are created via a borrow expression $\eref{\omega}{r}{p}$ that has an initial concrete provenance $r$ (abstract provenances are just used for types of function parameters). References are used in conjunction with place expressions $p$ that are places whose paths contain dereferences. For example, this program creates, reborrows, and mutates a reference:
\begin{align*}
    &\msfb{letprov} \langle r_1, r_2, r_3, r_4 \rangle \\
    &\msfb{let}~x : (\uty, \uty) = (0, 0); \\
    &\msfb{let}~y : \eref{\uniq}{r_2}{(\uty, \uty)} = \eref{\uniq}{r_1}{x}; \\
    &\msfb{let}~z : \eref{\uniq}{r_4}{\uty} = \eref{\uniq}{r_3}{(\ast y).1}; \\
    &{\ast z} := 1_\loc
\end{align*}
Consider the information flow induced by $\ast z := 1_\loc$. We need to compute all places that $z$ could point-to, in this case $x.1$, so $\loc$ can be added to the conflicts of $x.1$. Essentially, we must perform a \textit{pointer analysis}\,\cite{smaragdakis2015pointer}. 

The key idea is that Oxide already does a pointer analysis! Performing one is an essential task in ensuring ownership-safety. All we have to do is extract the relevant information with Oxide's existing judgments. This is represented by the information flow extension to the reference-mutation rule:
\vspace{4pt}
\begin{mathpar}
\typerule{T-AssignDeref}
  {
    \tcnew{\Sigma}{\Delta}{\Gamma}{\hlight{\Theta}}{e}{\tyvar{si}_n \hlight{\withslice{\kappa}}}{\Gamma_1}{\hlight{\Theta_1}} 
    \\
    \gray{\Delta; \Gamma_1 \vdash_\uniq{} p : \tyvar{si}_o }
    \\
    \Delta; \Gamma_1 \vdash_\uniq{} p \Rightarrow \loanset 
    \\
    \gray{\Delta; \Gamma_1 \vdash \tyvar{si}_n \lesssim \tyvar{si}_o \Rightarrow \Gamma' }
    \\
    \hlight{\Theta_2 = \Theta_1[\forall ~ ^\omega p' \in \loanset ~ . ~ \updateconflicts{\Theta_1}{p'}{\kappa}]}
  }
  {
    \tcnew{\Sigma}{\Delta}{\Gamma}{\hlight{\Theta}}{p := e}{\unitty \hlight{\withslice{\varnothing}}}{\Gamma' \vartriangleright p}{\hlight{\Theta_2}}
  }{tassignderef}
\end{mathpar}

Here, the important concept is Oxide's ownership safety judgment: $\Delta; \Gamma \vdash_\omega p \Rightarrow \loanset$, read as ``in the contexts $\Delta$ and $\Gamma$, $p$ can be used $\omega$-ly and points to a loan in $\loanset$.'' A loan $l ::= {^\omega p}$ is a place expression with an ownership-qualifier. In Oxide, this judgment is used to ensure that a place is used safely at a given level of mutability. For instance, in the example at the top of this column, if $\ast z := 1$ was replaced with $x.1 := 1$, then this would violate ownership-safety because $x$ is already borrowed by $y$ and $z$. 

In the example as written, the ownership-safety judgment for $\ast z$ would compute the loan set:
$$\loanset = \{\, {^\uniq(\ast z)}, {^\uniq (\ast y).1}, {^\uniq x.1}\}$$
Note that $x.1$ is in the loan set of $\ast z$. That suggests the loan set can be used as a pointer analysis. The complete details of computing the loan set can be found in \citet[p.~12]{weiss2019oxide}, but the summary for this example is:
\begin{enumerate}[leftmargin=*]
    \item Checking the borrow expression ``$\eref{\uniq}{r_1}{x}$'' gets the loan set for $x$, which is just $\{{^\uniq x}\}$, and so sets $\Gamma(r_1) = \{{^\uniq x}\}$.
    \item  Checking the assignment ``$y = \eref{\uniq}{r_1}{x}$'' requires that $\eref{\uniq}{r_1}{(\uty, \uty)}$ is a subtype of $\eref{\uniq}{r_2}{(\uty, \uty)}$, which requires that $r_1$ ``outlives'' $r_2$, denoted $r_1 :> r_2$.
    \item The constraint $r_1 :> r_2$ adds $\Gamma(r_1)$ to $\Gamma(r_2)$, so $\Gamma(r_2) = \{{^\uniq x}\}$.
    \item Checking ``$\eref{\uniq}{r_3}{(\ast y).1}$'' gets the loan set for $(\ast y).1$,  which is: $$\{{^\uniq p.1} \mid {^\uniq p} \in \Gamma(r_2)\} \cup \{{^\uniq (\ast y).1}\} = \{{^\uniq x.1}, {^\uniq (\ast y).1}\}$$ That is, the loans for $r_2$ are looked up in $\Gamma$ (to get $\{x\}$), and then the additional projection $\_.1$ is added on-top of each loan (to get $\{x.1\}$).
    \item Then $\Gamma(r_4) = \Gamma(r_3)$ because $r_3 :> r_4$.
    \item Finally, the loan set for $\ast z$ is: 
    $$\Gamma(r_4) \cup \{{^\uniq (\ast z)}\} = \{{^\uniq x.1}, {^\uniq (\ast y).1}, {^\uniq (\ast z)}\}$$
\end{enumerate}

Applying this concept to the \Cref{tr:tassignderef} rule, we compute information flow for reference-mutation as: when mutating $p$ with loans $\loanset$, add $\kappa_e$ to all the conflicts for every loan $^\uniq p' \in \loanset$.

\subsection{Function calls}
\label{sec:funcalls}

Finally, we examine how to modularly compute information flow through function calls, starting with syntax:
\begin{gather*}
\begin{aligned}
\msf{Type~Var}~\alpha \hspace{12pt} \msf{Frame~Var}~\varphi 
\end{aligned}
\\
\begin{aligned}
\bnfm{Expression}{e}{\ldots \mid f\langle\overline{\Phi}, \overline{\rho}, \overline{\tau}\rangle(\pi)} \\ 
\bnfm{Global~Entry}{\varepsilon}{\fndef} \\ 
\bnfm{Global~Env.}{\Sigma}{\bullet \mid \Sigma, \varepsilon}
\end{aligned}
\end{gather*}

Oxide functions are parameterized by frame variables $\varphi$ (for closures), abstract provenances $\varrho$ (for provenance polymorphism), and type variables $\alpha$ (for type polymorphism). Unlike Oxide, we restrict to functions with one argument for simplicity in the formalism. Calling a function $f$ requires an argument $\pi$ and any type-level parameters $\Phi, \rho$ and $\tau$.

The key question is: without inspecting its definition, what is the \textit{most precise} assumption we can make about a function's information flow while still being sound? By ``precise'' we mean ``if the analysis says there is a flow, then the flow actually exists'', and by ``sound'' we mean ``if a flow actually exists, then the analysis says that flow exists.'' For example consider this program:
\begin{align*}
    &\msfb{fn}~\msf{f}\langle \varrho_1, \varrho_2 \rangle(x: (\eref{\uniq}{\varrho_1}{\uty}, \eref{\shrd}{\varrho_2}{\uty})) \{ ~ \textcolor{gray}{\text{???}} ~ \}\\
    &\msfb{let}~x : \uty = 1_{\loc_1}; ~
    \msfb{let}~y : \uty = 2_{\loc_2}; \\
    &\msfb{letprov}\langle r_1, r_2 \rangle ~ \msfb{let}~t : (\eref{\uniq}{r_1}{\uty}, \eref{\shrd}{r_2}{\uty}) \\
    & \hspace{10pt} = ( \eref{\uniq}{r_1}{x}, \eref{\shrd}{r_2}{y}); \\
    &\msf{f}\langle r_1, r_2\rangle(t)
\end{align*}

First, what can $\msf{f}(t)$ mutate? Any data behind a shared reference is immutable, so only $\ast t.0$ could possibly be mutated, not $\ast t.1$. More generally, the argument's \textit{transitive mutable references} must be assumed to be mutated. 

Second, what are the inputs to the mutation of $\ast t.0$? This could theoretically be any possible value in the input, so both $\ast t.0$ and $\ast t.1$. More generally, every \textit{transitively readable place} from the argument must be assumed to be to be an input to the mutation. So in this example, a modular analysis of the information flow from calling $\msf{cp}$ would add $\{\loc_1, \loc_2\}$ to $\Theta(x)$ but not $\Theta(y)$.

To formalize these concepts, we first need to describe the transitive references of a place. The $\refs{\omega}{p, \tau}$ metafunction computes a place expression for every reference accessible from $p$. If $\omega = \uniq$ then this just includes unique references, otherwise it includes unique and shared ones.
\mathdef{refs}%
\begin{align*}
    \refs{\omega}{p, \tyvar{b}} &= \varnothing \\
    \refs{\omega}{p, (\tyvar{si}_1, \ldots, \tyvar{si}_n)} &= 
        \medcup_i \refs{\omega}{p.i, \tyvar{si}_i} \\ 
    \refs{\omega}{p, \eref{\omega'}{\rho}{\tyvar{xi}}} &= \begin{cases}
      \{\ast p\} \cup \refs{\omega}{\ast p, \tyvar{xi}} & \text{if $\omega \lesssim \omega'$} \\
      \varnothing & \text{otherwise}
    \end{cases}
\end{align*}

Here, $\omega \lesssim \omega'$ means ``a loan at $\omega$ can be used as a loan at $\omega'$'', defined as $\uniq \not\lesssim \shrd$ and $\omega \lesssim \omega'$ otherwise. Then $\loans{\omega}{p, \tau, \Delta, \Gamma}$ can be defined as the set of concrete places accessible from those transitive references:
\mathdef{loans}%
\begin{align*}
    &\loans{\omega}{p, \tau, \Delta, \Gamma} \eqdef \\ 
    &\hspace{10pt} \bigcup_{p_1 \in \refs{\omega}{p, \tau}} \{p_2 \mid {^\omega p_2} \in \loanset\} &&\text{where $\Delta; \Gamma \vdash_\omega p_1 \Rightarrow \loanset$}
\end{align*}\

Finally, the function application rule can be revised to include information flow as follows:
\vspace{6pt}
\begin{mathpar}
\typerule{T-App}
  {
    \gray{\overline{\Sigma; \Delta; \Gamma \vdash \Phi} }
    \\
    \gray{\overline{\Delta; \Gamma \vdash \rho} }
    \\
    \gray{\overline{\Sigma; \Delta; \Gamma \vdash \tyvar{si}}}
    \\\\
    \Sigma(f) = \fndef
    \\
    \tcnew{\Sigma}{\Delta}{\Gamma}{\hlight{\Theta}}{\pi}
    {\tyvar{si}_a
        \overline{[\sfrac{\Phi}{\varphi}]}
        \overline{[\sfrac{\rho}{\varrho}]}
        \overline{[\sfrac{\tyvar{si}}{\alpha}]}
     \hlight{\withslice{\kappa}}}
    {\Gamma_1}{\hlight{\Theta}}
    \\
    \gray{\Delta; \Gamma_1 \vdash \overline{\varrho_2\overline{[\sfrac{\rho}{\varrho}]} :> \varrho_1\overline{[\sfrac{\rho}{\varrho}]}} \Rightarrow \Gamma_2}
    \\
    \hlt{\kappa_\arrg = \kappa \cup \medcup_{p \in \loans{\shrd}{\pi, \tyvar{si}_a, \Delta, \Gamma_2}} \Theta(p)}
    \\
    \hlight{
    \begin{minipage}[b]{\linewidth}
    \begin{align*}
    \Theta' = \Theta[&\forall p \in \loans{\uniq}{\pi, \tyvar{si}_a, \Delta, \Gamma_2} ~ . 
    \\
    &\updateconflicts{\Theta}{p}{\kappa_\arrg}] 
    \end{align*}
    \end{minipage}
    }
  }
  {
    \tcnew{\Sigma}{\Delta}{\Gamma}{\hlight{\Theta}}{f\langle\overline{\Phi}, \overline{\rho}, \overline{\tyvar{si}}\rangle(\pi)}
    {\tyvar{si}_r 
        \overline{[\sfrac{\Phi}{\varphi}]}
        \overline{[\sfrac{\rho}{\varrho}]}
        \overline{[\sfrac{\tyvar{si}}{\alpha}]}
     \hlight{\withslice{\kappa_\arrg}}}
    {\Gamma_2}{\hlight{\Theta'}}
  }{tapp}
\end{mathpar}

The collective dependencies of the input $\pi$ are collected into $\kappa_\arrg$, and then every unique reference is updated with $\kappa_\arrg$. Additionally, the function's return value is assumed to be influenced by any input, and so has dependencies $\kappa_\arrg$. 

Note that this rule does not depend on the body $e$ of the function $f$, only its type signature in $\Sigma$. This is the key to the modular approximation. Additionally, it means that this analysis can trivially handle higher-order functions. If $f$ were a parameter to the function being analyzed, then no control-flow analysis is needed to guess its definition.

\section{Soundness}
\label{sec:soundness}

\begin{figure*}
\centering
\begin{minipage}{0.32\linewidth}
\begin{lstlisting}[xleftmargin=15pt]
fn get_count(
  h: &mut HashMap<String, u32>, 
  k: String
) -> u32 {
  if !h.contains_key(&k) {
    h.insert(k, 0); 0
  } else {
    *h.get(&k).unwrap()
  }
}
\end{lstlisting}
\end{minipage}%
\hspace{0.03\linewidth}
\begin{minipage}{0.6\linewidth}
\includegraphics[width=\linewidth]{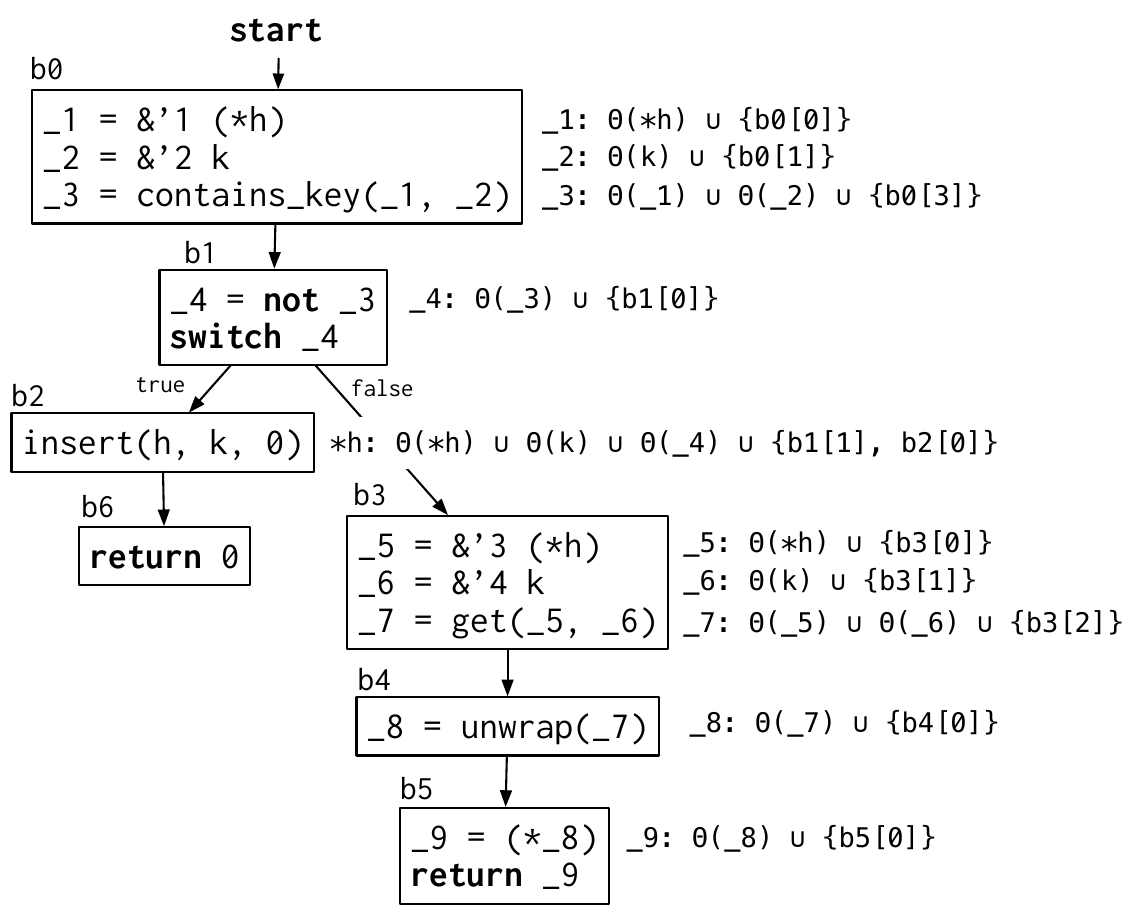}
\end{minipage}%
\cprotect\caption{Example of how \sysname{} computes information flow. On the left is a Rust function \verb|get_count| that finds a value in a hash map for a given key, and inserts 0 if none exists. On the right \verb|get_count| is lowered into Rust's MIR control-flow graph, annotated with information flow. Each rectangle is a basic block, named at the top. Arrows indicate control flow (panics omitted).
Beside each instruction is the result of the information flow analysis, which maps place expressions to locations in the CFG (akin to $\Theta$ in \Cref{sec:analysis}). For example, the \verb|insert| function call adds dependencies to \verb|*h| because it is assumed to be mutated, since it is a mutable reference. Additionally, the \verb|switch| instructions and \verb|_4| variable are added as dependencies to \verb|h| because the call to \verb|insert| is control-dependent on the switch.}
\label{fig:mir_example}
\end{figure*}

To characterize the correctness of our analysis, we seek to prove its \textit{soundness}: if a true information flow exists in a program, then the analysis computes that flow. The standard soundness theorem for information flow systems is \textit{noninterference}\,\cite{goguen1982security}. At a high level, noninterference states that for a given program and its dependencies, and for any two execution contexts, if the dependencies are equal between contexts, then the program will execute with the same observable behavior in both cases. For our analysis, we focus just on values produced by the program, instead of other behaviors like termination or timing.

To formally define noninterference within Oxide, we first need to explore Oxide's operational semantics. Oxide programs are executed in the context of a stack of frames that map variables to values:
\begin{gather*}
\begin{aligned}
\bnfm{Stack}{\sigma}{\bullet \mid \sigma \stacksep \varsigma} \\
\bnfm{Stack~Frame}{\varsigma}{\bullet \mid \varsigma, x \mapsto v} \\
\end{aligned}
\end{gather*}

For example, in the empty stack $\bullet$, the expression ``$\olet{x}{\uty}{1}{x := 2}$'' would first add $x \mapsto 1$ to the stack. Then executing $x := 2$ would update $\sigma(x) = 2$. More generally, we use the shorthand $\sigma(p)$ to mean ``reduce $p$ to a concrete location $\pi$, then look up the value of $\pi$ in $\sigma$.''

The key ingredient for noninterference is the equivalence of dependencies between stacks. That is, for two stacks $\sigma_1$ and $\sigma_2$ and a set of dependencies $\kappa$ in a context $\Theta$, we say those stacks are \textit{the same up to $\kappa$} if all $p$ with $\Theta(p) \subseteq \kappa$ are the same between stacks. Formally, the dependencies of $\kappa$ and equivalence of heaps are defined as:
\mathdef{deps}%
\mathdef{stackeq}%
\mathdef{stackdepeq}%
\begin{align*}
    \deps{\Theta}{\kappa} &\eqdef \{p \mid p \mapsto \kappa_p \in \Theta \wedge \kappa_p \subseteq \kappa\} \\
    \sigma_1 \stackeq{P} \sigma_2 &\eqdef \forall p \in P ~ . ~ \sigma_1(p) = \sigma_2(p) \\
    \sigma_1 \stackdepeq{\Theta}{\kappa} \sigma_2 &\eqdef \sigma_1 \stackeq{\deps{\Theta}{\kappa}} \sigma_2
\end{align*}

Then we define noninterference as follows:

\savetheorem{noninterference}
\begin{theorem}[Noninterference]
\label{thm:noninterference}
Let $e$ such that: 
$$\tcnew{\Sigma}{\bullet}{\Gamma}{\Theta}{e}{\tau \withslice{\kappa}}{\Gamma'}{\Theta'}$$
For $i \in \{1, 2\}$, let $\sigma_i$ such that: 
$$\Sigma \vdash \sigma_i : \Gamma \hspace{12pt} \text{and} \hspace{12pt} \evalsto{\sigma_i}{e}{\stepped{\sigma}_i}{v_i}$$
Then:
\begin{enumerate}[(a)]
    \item $\sigma_1 \stackdepeq{\Theta}{\kappa} \sigma_2 \implies v_1 = v_2$
    \item $\forall p \mapsto \kappa_p \in \Theta' ~ . ~ \sigma_1 \stackdepeq{\Theta}{\kappa_p} \sigma_2 \implies \stepped{\sigma}_1(p) = \stepped{\sigma}_2(p)$
\end{enumerate}
\end{theorem}

This theorem states that given a well-typed expression $e$ and corresponding stacks $\sigma_i$, then its output $v_i$ should be equal if the expression's dependencies $\kappa$ are initially equal. Moreover, for any place expression $p$, if its dependencies in the output context $\Theta'$ are initially equal then the stack value will be the same after execution.

Note that the context $\Delta$ is required to be empty because an expression $e$ can only evaluate if it does not contain abstract type or provenance variables. The judgment $\Sigma \vdash \sigma_i : \Gamma$ means ``the stack $\sigma_i$ is well-typed under $\Sigma$ and $\Gamma$''. That is, for all places $\pi$ in $\Gamma$, then $\pi \in \sigma$ and $\sigma(\pi)$ has type $\Gamma(\pi)$.

The proof of \Cref{thm:noninterference}, found in \Cref{sec:appendix_proofs}, guarantees that we can soundly compute information flow for Oxide programs.

\section{Implementation}
\label{sec:implementation}

\noindent Our formal model provides a sound theoretical basis for analyzing information flow in Oxide. However, Rust is a more complex language than Oxide, and the Rust compiler uses many intermediate representations beyond its surface syntax. Therefore in this section, we describe the key details of how our system, \sysname{}, bridges theory to practice. Specifically: 
\begin{enumerate}[leftmargin=*]
    \item Rust computes lifetime-related information on a control-flow graph (CFG) program representation, not the high-level AST. So we translate our analysis to work for CFGs (\Cref{sec:mir}).
    \item Rust does not compute the loan set for lifetimes directly like in Oxide. So we must reconstruct the loan sets given the information exported by Rust (\Cref{sec:lifetimes}).
    \item Rust contains escape hatches for ownership-unsafe code that cannot be analyzed using our analysis. So we describe the situations in which our analysis is unsound for Rust programs (\Cref{sec:limitations}).
\end{enumerate}

\subsection{Analyzing control-flow graphs}
\label{sec:mir}

The Rust compiler lowers programs into a ``mid-level representation'', or MIR, that represents programs as a control-flow graph. Essentially, expressions are flattened into sequences of instructions (basic blocks) which terminate in instructions that can jump to other blocks, like a branch or function call.  \Cref{fig:mir_example} shows an example CFG and its information flow.

To implement the modular information flow analysis for MIR, we reused standard static analysis techniques for CFGs, i.e., a flow-sensitive, forward dataflow analysis pass where:

\begin{itemize}[leftmargin=*]
    \item At each instruction, we maintain a mapping from place expressions to a set of locations in the CFG on which the place is dependent, comparable to $\Theta$ in \Cref{sec:analysis}. 
    \item A transfer function updates $\Theta$ for each instruction, e.g. $p := e$ follows the same rules as in \Cref{tr:tassignderef} by adding the dependencies of $e$ to all conflicts of aliases of $p$.
    \item The input $\Theta^\msf{in}$ to a basic block is the join of each of the output $\Theta^\msf{out}_i$ for each incoming edge, i.e. $\Theta^\msf{in} = \bigvee_i \Theta^\msf{out}_i$ . The join operation is key-wise set union, or more precisely:
    \[\Theta_1 \vee \Theta_2 \eqdef \{x \mapsto \Theta_1(x) \cup \Theta_2(x) \mid x \in \Theta_1 \vee x \in \Theta_2\}\]
    \item We iterate this analysis to a fixpoint, which we are guaranteed to reach because $\langle\Theta, \vee\rangle$ forms a join-semilattice.
\end{itemize}

To handle indirect information flows via control flow, such as the dependence of \verb|h| on \verb|contains_key| in \Cref{fig:mir_example}, we compute the control-dependence between instructions.  We define control-dependence following \citet{ferrante1987program}:  an instruction $X$ is control-dependent on $Y$ if there exists a directed path $P$ from $X$ to $Y$ such that any $Z$ in $P$ is post-dominated by $Y$, and $X$ is not post-dominated by $Y$. An instruction $X$ is post-dominated by $Y$ if $Y$ is on every path from $X$ to a \Verb|return| node. We compute control-dependencies by generating the post-dominator tree and frontier of the CFG using the algorithms of \citet{cooper2001simple} and \citet{cytron1989efficient}, respectively.

Besides a return, the only other control-flow path out of a function in Rust is a panic. For example, each function call in \Cref{fig:mir_example} actually has an implicit edge to a panic node (not depicted). Unlike exceptions in other languages, panics are designed to indicate unrecoverable failure. Therefore we exclude panics from our control-dependence analysis.

\subsection{Computing loan sets from lifetimes}
\label{sec:lifetimes}

To verify ownership-safety (perform ``borrow-checking''), the Rust compiler does not explicitly build the loan sets of lifetimes (or provenances in Oxide terminology). The borrow checking algorithm performs a sort of flow-sensitive dataflow analysis that determines the range of code during which a lifetime is valid, and then checks for conflicts e.g. in overlapping lifetimes (see the non-lexical lifetimes RFC\,\cite{nllrfc}).

However, Rust's borrow checker relies on the same fundamental language feature as Oxide to verify ownership-safety: outlives-constraints. For a given Rust function, Rust can output the set of outlives-constraints between all lifetimes in the function. These lifetimes are generated in the same manner as in Oxide, such as from inferred subtyping requirements or user-provided outlives-constraints. Then given these constraints, we compute loan sets via a process similar to the ownership-safety judgment described in \Cref{sec:references}. In short, for all instances of borrow expressions $\eref{\omega}{r}{p}$ in the MIR program, we initialize $\Gamma(r) = \{p\}$. Then we propagate loans via $\Gamma(r) = \medcup_{r' :> r} \Gamma(r')$ until $\Gamma$ reaches a fixpoint.

\subsection{Handling ownership-unsafe code}
\label{sec:limitations}


Rust has a concept of \textit{raw pointers} whose behavior is comparable to pointers in C. For a type \verb|T|, an immutable reference has type \verb|&T|, while an immutable raw pointer has type \verb|*const T|. Raw pointers are not subject to ownership restrictions, and they can only be used in blocks of code demarcated as \verb|unsafe|. They are primarily used to interoperate with other languages like C, and to implement primitives that cannot be proved as ownership-safe via Rust's rules.

Our pointer and mutation analysis fundamentally relies on ownership-safety for soundness. We do not try to analyze information flowing directly through unsafe code, as it would be subject to the same difficulties of C++ in \Cref{sec:intro}. While this limits the applicability of our analysis, empirical studies have shown that most Rust code does not (directly) use unsafe blocks\,\cite{astrauskas2020programmers,evans2020rust}. We further discuss the impact and potential mitigations of this limitation in \Cref{sec:discussion}.

\section{Evaluation}
\label{sec:evaluation}

\begin{table*}
\begin{tabular}{l|l|l|r|r|r|r}
    \textbf{Project} & \textbf{Crate} & \textbf{Purpose} & \textbf{LOC} & \textbf{\# Vars} & \textbf{\# Funcs} & \textbf{Avg. Instrs/Func} \\ \hline
    \href{https://github.com/rayon-rs/rayon}{rayon} &  & Data parallelism library & 15,524 & 10,607 & 1,079 & 16.6 \\ \hline
    \href{https://github.com/SergioBenitez/Rocket}{Rocket} & core/lib & Web backend framework & 10,688 & 12,040 & 741 & 25.5 \\ \hline
    \href{https://github.com/ctz/rustls}{rustls} & rustls & TLS implementation & 16,866 & 23,407 & 868 & 42.4 \\ \hline
    \href{https://github.com/mozilla/sccache}{sccache} &  & Distributed build cache & 23,202 & 23,987 & 643 & 62.1 \\ \hline
    \href{https://github.com/dimforge/nalgebra}{nalgebra} &  & Numerics library & 31,951 & 35,886 & 1,785 & 26.7 \\ \hline
    \href{https://github.com/image-rs/image}{image} &  & Image processing library & 20,722 & 39,077 & 1,096 & 56.8 \\ \hline
    \href{https://github.com/hyperium/hyper}{hyper} &  & HTTP server & 15,082 & 44,900 & 790 & 82.9 \\ \hline
    \href{https://github.com/mrDIMAS/rg3d}{rg3d} &  & 3D game engine & 54,426 & 59,590 & 3,448 & 25.7 \\ \hline
    \href{https://github.com/xiph/rav1e}{rav1e} &  & Video encoder & 50,294 & 76,749 & 931 & 115.4 \\ \hline
    \href{https://github.com/RustPython/RustPython}{RustPython} & vm & Python interpreter & 47,927 & 97,637 & 3,315 & 51.0 \\ \hline
    \multicolumn{1}{l}{} & 
    \multicolumn{1}{l}{} &
    \multicolumn{1}{r|}{\textbf{Total:}} &
    286,682 & 435,979 & 14,696
\end{tabular}
\caption{Dataset of crates used to evaluate information flow precision, ordered in increasing number of variables analyzed. Each project often contains many crates, so a sub-crate is specified where applicable, and the root crate is analyzed otherwise. Metrics displayed are LOC (lines of code), number of variables, number of functions, and the average number of MIR instructions per function (size of CFG).}
\label{tab:dataset}
\end{table*}

\Cref{sec:soundness} established that our analysis is \textit{sound}. The next question is whether it is \textit{precise}: how many spurious flows are included by our analysis? 
We evaluate two directions:

\begin{enumerate}[leftmargin=*]
    \item What if the analysis had \textit{more} information? If we could analyze the definitions of called functions, how much more precise are whole-program flows vs. modular flows?
    \item What if the analysis had \textit{less} information? If Rust's type system was more like C++, i.e. lacking ownership, then how much less precise do the modular flows become?
\end{enumerate}

\noindent To answer these questions, we created three modifications to Flowistry:

\begin{itemize}[leftmargin=*]
    \item \wholeprogram{}: the analysis recursively analyzes information flow within the definitions of called functions. For example, if calling a function \lstinline|f(&mut x, y)| where \verb|f| does not actually modify \verb|x|, then the \wholeprogram{} analysis will not register a flow from \verb|y| to \verb|x|.
    \item \mutblind{}: the analysis does not distinguish between mutable and immutable references. For example, if calling a function \verb|f(&x)|, then the analysis assumes that \verb|x| can be modified.
    \item \pointerblind{}: the analysis does not use lifetimes to reason about references, and rather assumes all references of the same type can alias. For example, if a function takes as input \lstinline|f(x: &mut i32, y: &mut i32)| then \verb|x| and \verb|y| are assumed to be aliases.
\end{itemize}

The \wholeprogram{} modification represents the most precise information flow analysis we can feasibly implement. 
The \mutblind{} and \pointerblind{} modifications represent an ablation of the precision provided by ownership types.
Each modification can be combined with the others, representing $2^3 = 8$ possible conditions for evaluation.

To better understand \wholeprogram{}, say we are analyzing the information flow for an expression \lstinline|f(&mut x, y)| where \verb|f| is defined as \lstinline|f(a, b) { (*a).1 = b; }|. After analyzing the definition of \verb|f|, we translate flows to parameters of \verb|f| into flows on arguments of the call to \verb|f|. So the flow \lstinline[mathescape]|b $\rightarrow$ (*a).1| is translated into \lstinline[mathescape]|y $\rightarrow$ x.1|.
\
Additionally, if the definition of \verb|f| is not available, then we fall back to the modular analysis. Importantly, due to the architecture of the Rust compiler, the only available definitions are those \textit{within the package being analyzed}. Therefore even with \wholeprogram{}, we cannot recurse into e.g. the standard library.

With these three modifications, we compare the number of flows computed from a dataset of Rust projects (\Cref{sec:dataset}) to quantitatively (\Cref{sec:quant}) and qualitatively (\Cref{sec:qual}) evaluate the precision of our analysis.

\subsection{Dataset}
\label{sec:dataset}

To empirically compare these modifications, we curated a dataset of Rust packages (or ``crates'') to analyze. We had two selection criteria:
\begin{enumerate}[leftmargin=*]
    \item To mitigate the single-crate limitation of \wholeprogram{}, we preferred large crates so as to see a greater impact from the \wholeprogram{} modification. We only considered crates with over 10,000 lines of code as measured by the \Verb|cloc| utility\,\cite{cloc}.
    \item To control for code styles specific to individual applications, we wanted crates from a wide range of domains.
\end{enumerate}

After a manual review of large crates in the Rust ecosystem, we selected 10 crates, shown in \Cref{tab:dataset}. We built each crate with as many feature flags enabled as would work on our Ubuntu 16.04 machine. Details like the specific flags and commit hashes can be found in \Cref{sec:eval_details}.

For each crate, we ran the information flow analysis on every function in the crate, repeated under each of the 8 conditions. Within a function, for each local variable $x$, we compute the size of $\Theta(x)$ at the exit of the CFG --- in terms of program slicing, we compute the size of the variable's backward slice at the function's return instructions. The resulting dataset then has four independent variables (crate, function, condition, variable name) and one dependent variable (size of dependency set) for a total of 3,487,832 data points.

Our main goal in this evaluation is to analyze precision, not performance. Our baseline implementation is reasonably optimized --- the median per-function execution time was $370.24\mu\text{s}$. But \wholeprogram{} is designed to be as precise as possible, so its naive recursion is sometimes extremely slow. For example, when analyzing the \verb|GameEngine::render| function of the \verb|rg3d| crate (with thousands of functions in its call graph), the modular analysis takes 0.13s while the recursive analysis takes 23.18s, a $178\times$ slowdown. Future work could compare our modular analysis to whole-program analyses across the precision/performance spectrum, such as in the extensive literature on context-sensitivity\,\cite{smaragdakis2015pointer}. 

\subsection{Quantitative results}
\label{sec:quant}

\begin{figure}[t]
\includegraphics[width=\linewidth]{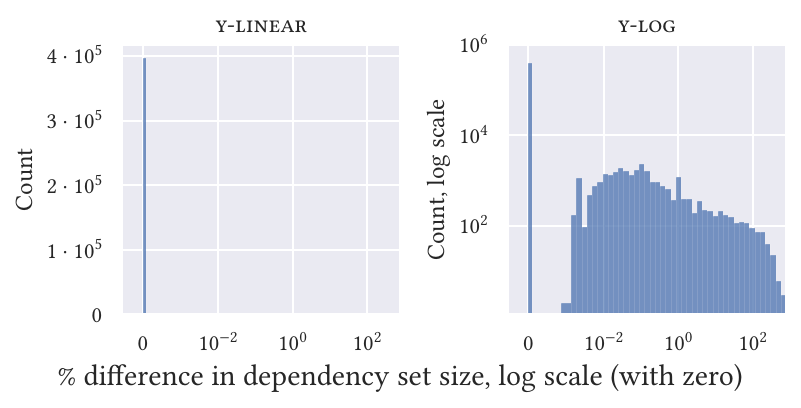}
\caption{Distribution in differences of dependency set size between \wholeprogram{} and \baseline{} analyses. The x-axis is a log-scale with 0 added for comparison. Most sets are the same, so 0 dominates (left). A log-scale (right) shows the tail more clearly.}
\label{fig:recurse}
\end{figure}

We observed no meaningful patterns from the interaction of modifications --- for example, in a linear regression of the interaction of \mutblind{} and \pointerblind{} against the size of the dependency set, each condition is individually statistically significant ($p < 0.001$) while their interaction is not ($p = 0.337$). So to simplify our presentation, we focus only on four conditions: three for each modification individually active with the rest disabled, and one for all modifications disabled, referred to as \baseline{}. 

\subsubsection{\wholeprogram{}}
\label{sec:whole}

For \wholeprogram{}, we compare against \baseline{} to answer our first evaluation question: how much more precise is a whole-program analysis than a modular one? To quantify precision, we compare the \textit{percentage increase in size} of dependency sets for a given variable between two conditions. For instance, if \wholeprogram{} computes $|\Theta(x)| = 2$ and $\baseline{}$ computes $|\Theta(x)| = 5$ for some $x$, then the difference is $(5 - 2) / 2 = 1.5 = 150\%$. 

\Cref{fig:recurse} shows a histogram of the differences between \wholeprogram{} and \baseline{} for all variables. In 94\% of all cases, the \wholeprogram{} and \baseline{} conditions produce the same result and hence have a difference of 0. In the remaining 6\% of cases with a non-zero difference, visually enhanced with a log-scale in \Cref{fig:recurse}-right, the metric follows a right-tailed log-normal distribution. We can summarize the log-normal by computing its median, which is 7\%. This means that within the 6\% of non-zero cases, the median difference is an increase in size by 7\%. Thus, the modular approximation does not significantly increase the size of dependency sets in the vast majority of cases.

\subsubsection{\mutblind{} and \pointerblind{}}
\label{sec:blind}

\begin{figure}[t]
\includegraphics[width=\linewidth]{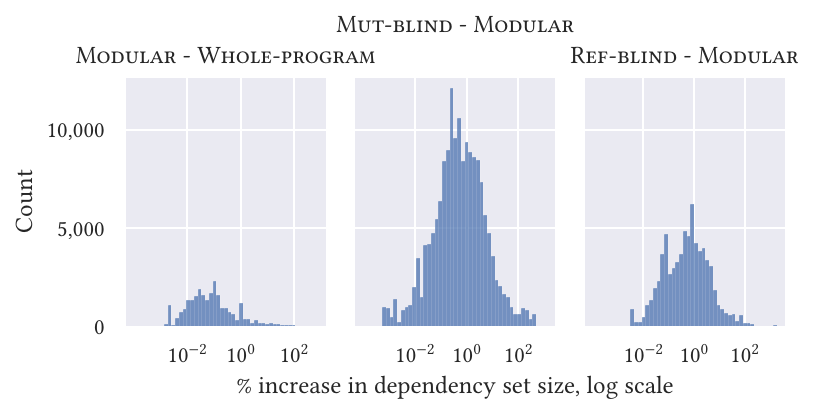}
\caption{Distribution in differences between \baseline{} and each alternative condition, with zeros excluded to highlight the shape of each distribution. \mutblind{} and \pointerblind{} both reduce the precision more often and more severely than \baseline{} does vs. \wholeprogram{}.}
\label{fig:all_dist}
\end{figure}

Next, we address our second evaluation question: how much less precise is an analysis with weaker assumptions about the program than the \baseline{} analysis? For this question, we compare the size of dependency sets between the \mutblind{} and \pointerblind{} conditions versus \baseline{}. \Cref{fig:all_dist} shows the corresponding histograms of differences, with the \wholeprogram{} vs. \baseline{} histogram included for comparison.

First, the \mutblind{} and \pointerblind{} modifications reduce the precision of the analysis more often and with a greater magnitude than \baseline{}{} does vs. \wholeprogram{}. 39\% of \mutblind{} cases and 17\% of \pointerblind{} cases have a non-zero difference. Of those cases, the median difference in size is 50\% for \mutblind{} and 56\% for \pointerblind{}.

Therefore, the information from ownership types is valuable in increasing the precision of our information flow analysis. Dependency sets are often larger without access to information about mutability or lifetimes.

\subsection{Qualitative results}
\label{sec:qual}

\begin{figure*}
\includegraphics[width=0.9\linewidth]{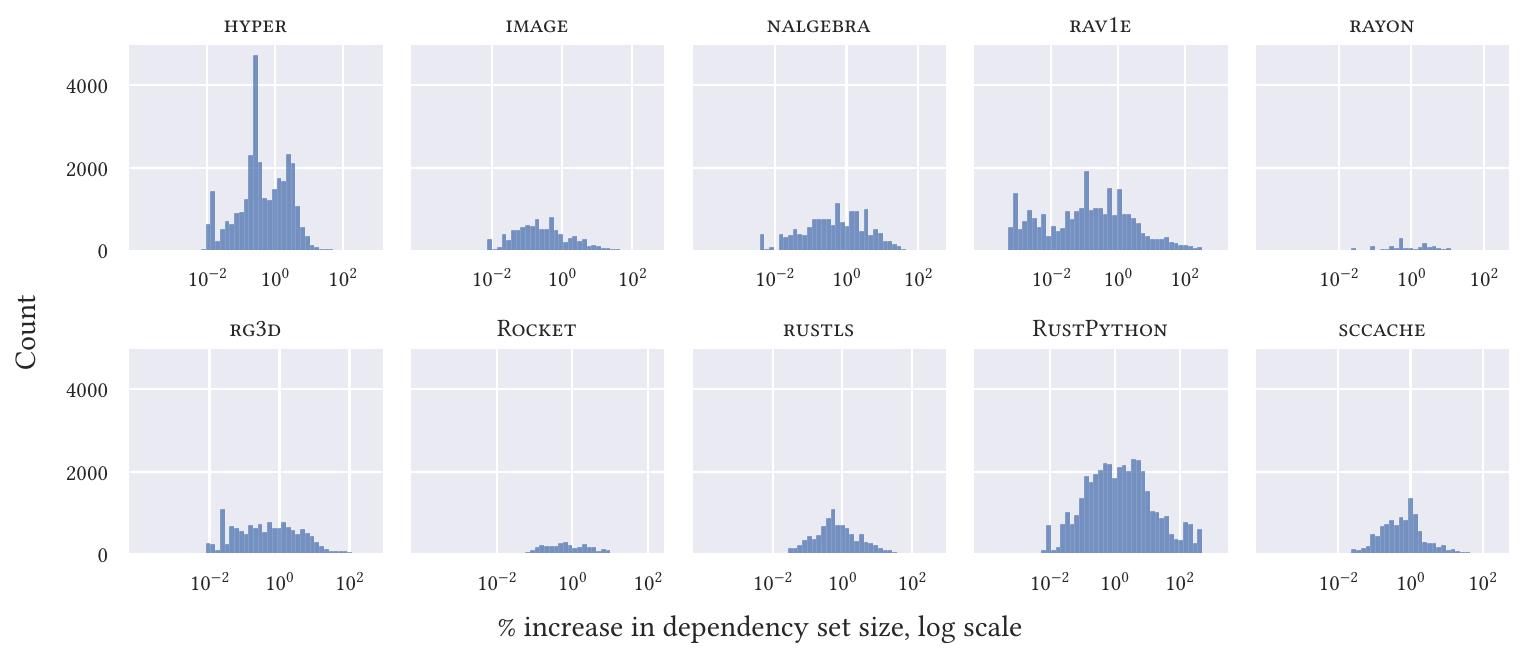}
\caption{Distribution of non-zero differences between \baseline{} and \mutblind{}, broken down by crate. }
\label{fig:crates}
\end{figure*}

The statistics convey a sense of how often each condition influences precision. But it is equally valuable to understand the kind of code that leads to such differences. For each condition, we manually inspected a sample of cases with non-zero differences vs. \baseline{}.

\subsubsection{Modularity.}\label{sec:whole-vs-mod}

One common source of imprecision in modular flows is when functions take a mutable reference as input for the purposes of passing the mutable permission off to an element of the input. 

\begin{lstlisting}
fn crop<I: GenericImageView>(
  image: &mut I, x: u32, y: u32, 
  width: u32, height: u32
) -> SubImage<&mut I> {
  let (x, y, width, height) = 
    crop_dimms(image, x, y, width, height);
  SubImage::new(image, x, y, width, height)
}
\end{lstlisting}

For example, the function \Verb|crop| from the \Verb|image| crate returns a mutable view on an image. No data is mutated, only the mutable permission is passed from whole image to sub-image. However, a modular analysis on the \Verb|image| input would assume that \Verb|image| is mutated by \Verb|crop|.

Another common case is when a value only depends on a subset of a function's inputs. The modular approximation assumes all inputs are relevant to all possible mutations, but this is naturally not always the case.

\begin{lstlisting}
fn solve_lower_triangular_with_diag_mut<R2,C2,S2>(
  &self, b: &mut Matrix<N, R2, C2, S2>, diag: N,
) -> bool {
  if diag.is_zero() { return false; }
  // logic mutating b...
  true
}
\end{lstlisting}

For example, this function from \Verb|nalgebra| returns a boolean whose value solely depends on the argument \Verb|diag|. However, a modular analysis of a call to this function would assume that \lstinline|self| and \verb|b| is relevant to the return value as well.

\subsubsection{Mutability}\label{sec:mut}

The reason \mutblind{} is less precise than \baseline{} is quite simple --- many functions take immutable references as inputs, and so many more mutations have to be
assumed.

\begin{lstlisting}
fn read_until<R, F>(io: &mut R, func: F)
  -> io::Result<Vec<u8>> 
  where R: Read, F: Fn(&[u8]) -> bool
{
  let mut buf = vec![0; 8192]; let mut pos = 0;
  loop {
    let n = io.read(&mut buf[pos..])?; pos += n;
    if func(&buf[..pos]) { break; } 
    // ...
  }
}
\end{lstlisting}

For instance, this function from \Verb|hyper| repeatedly calls an input function \rust|func| with segments of an input buffer. Without a control-flow analysis, it is impossible to know what functions \verb|read_until| will be called with. And so \mutblind{} must always assume that \rust|func| could mutate \rust|buf|. However, \baseline{} can rely on the immutability of shared references and deduce that \rust|func| could not mutate \rust|buf|.

\subsubsection{Lifetimes}\label{sec:life}

Without lifetimes, our analysis has to make more conservative assumptions about objects that could possibly alias. We observed many cases in the \pointerblind{} condition where two references shared different lifetimes but the same type, and so had to be classified as aliases.

\begin{lstlisting}
fn link_child_with_parent_component(
  parent: &mut FbxComponent,
  child: &mut FbxComponent,
  child_handle: Handle<FbxComponent>,
) { match parent {
  FbxComponent::Model(model) => {
    model.geoms.push(child_handle),
  },
  // ..
}}
\end{lstlisting}

For example, the \verb|link_child_with_parent_component| function in \verb|rg3d| takes mutable references to a \verb|parent| and \verb|child|. These references are guaranteed not to alias by the rules of ownership, but a naive pointer analysis must assume they could, so modifying \verb|parent| could modify \verb|child|.

\newcommand{\imgframe}[1]{\fcolorbox{gray}{white}{#1}}

\subsection{Threats to validity}
\begin{figure*}
\begin{subfigure}{0.48\linewidth}
\imgframe{\includegraphics[width=0.9\linewidth]{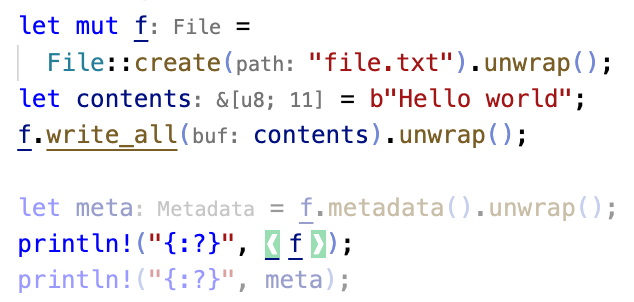}}
\ \\ \\
\imgframe{\includegraphics[width=0.9\linewidth]{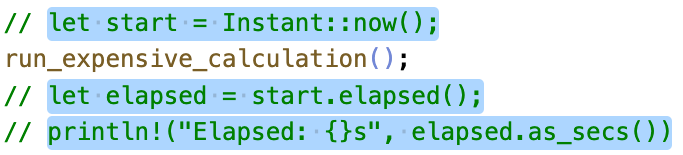}}
\cprotect\caption{A program slicer integrated into VSCode. Above, the user selects a slicing criterion like the variable \verb|f|. Then the slicer highlights the criterion in green, and fades out lines that are not part of the backward slice on \verb|f|. For example, \verb|write_all| mutates the file so it is in the slice, while \verb|metadata| reads the file so it is not in the slice.

Below, the user can manipulate aspects of a program such as commenting out code related to timing. The user computes a forward slice on \verb|start|, adds this slice to their selection (in blue), then tells the IDE to comments out all lines in the selection.}
\label{fig:slicer}
\end{subfigure}
\hspace{0.01\linewidth}
\begin{subfigure}{0.48\linewidth}
\imgframe{\includegraphics[width=\linewidth]{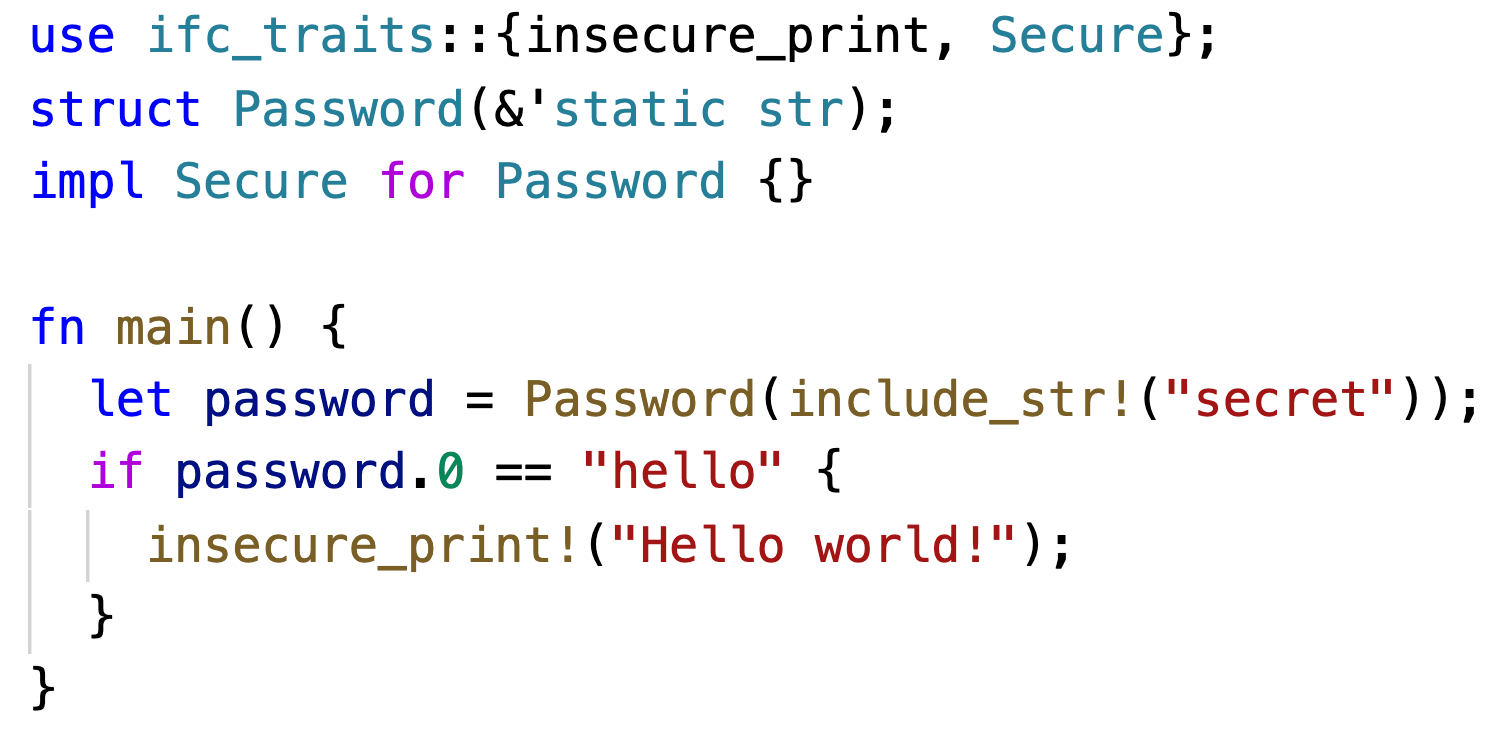}}
\ \\
\imgframe{\includegraphics[width=\linewidth]{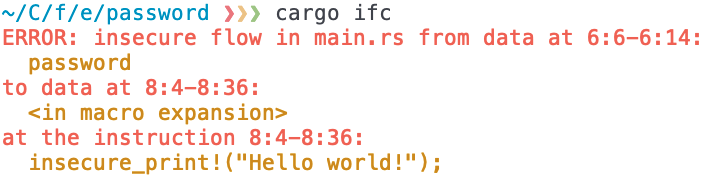}}
\cprotect\caption{An IFC checker. Above, the \verb|ifc_traits| library exports a \verb|Secure| for users to mark sensitive data, like \verb|Password|, and insecure operations like \verb|insecure_print|. Below, a compiler plugin invoked on the program checks for information flows from data with a type implementing \verb|Secure| to insecure operations. Here \verb|insecure_print| is conditionally executed based on a read from \verb|PASSWORD|, so this flow is flagged.}
\label{fig:ifc}
\end{subfigure}
\caption{Two applications of information flow built using \sysname{}.}
\label{fig:apps}
\end{figure*}

Finally, we address the issue: how meaningful are the results above? How likely would they generalize to arbitrary code rather than just our selected dataset? We discuss a few threats to validity below.

\subsubsection{Are the results due to only a few crates?}

If differences between techniques only arose in a small number of situations that happen to be in our dataset, then our technique would not be as generally applicable. To determine the variation between crates, we generated a histogram of non-zero differences for the \baseline{} vs. \mutblind{} comparison, broken down by crate in \Cref{fig:crates}. 

As expected, the larger code bases (e.g. rav1e and RustPython) have more non-zero differences than smaller codebases --- in general the correlation between non-zero differences and total number of variables analyzed is strong, $R^2 = 0.79$. However variation also exists for crates with roughly the same number of variables like \verb|image| and \verb|hyper|. \mutblind{} reduces precision for variables in \verb|hyper| more often than \verb|image|. A qualitative inspection of the respective codebases suggests this may be because \verb|hyper| simply makes greater use of immutable references in its API.

These findings suggest that the impact of ownership types and the modular approximation likely do vary with code style, but a broader trend is still observable across all code.

\subsubsection{Would \wholeprogram{} be more precise with access to dependencies?}

A limitation of our whole-program analysis is our inability to access function definitions outside the current crate. Without this limitation, it may be that the \baseline{} analysis would be significantly worse than \wholeprogram{}. So for each variable analyzed by \wholeprogram{}, we additionally computed whether the information flow for that variable involved a function call across a crate boundary.

Overall 96\% of cases reached at least one crate boundary, suggesting that this limitation does occur quite often in practice. However, the impact of the limitation is less clear. Of the 96\% of cases that hit a crate boundary, 6.6\% had a non-zero difference between \baseline{} and \wholeprogram{}. Of the 4\% that did not hit a crate boundary, 0.6\% had a non-zero difference. One would expect that \wholeprogram{} would be the most precise when the whole program is available (no boundary), but instead it was much closer to \baseline{}. 

Ultimately it is not clear how much more precise \wholeprogram{} would be given access to all a crate's dependencies, but it would not necessarily be a significant improvement over the benchmark presented.

\subsubsection{Is ownership actually important for precision?}

The finding that \pointerblind{} only makes a difference in 17\% of cases may seem surprisingly small. For instance, \citet{shapiro1997effects} found in a empirical study of slices on C programs that ``using a pointer analysis with an average points-to set size twice as large as a more precise pointer analysis led to an increase of about 70\% in the size of [slices].''

A  limitation of our ablation is that the analyzed programs were written to satisfy Rust's ownership safety rules. Disabling lifetimes does not change the structure of the programs to become more C-like --- Rust generally encourages a code style with fewer aliases to avoid dealing with lifetimes. A fairer comparison would be to implement an application idiomatically in both Rust and Rust-but-without-feature-X, but such an evaluation is not practical. It is therefore likely that our results understate the true impact of ownership types on precision given this limitation.

\section{Applications}
\label{sec:applications}

We have demonstrated that ownership can be leveraged to build an information flow analysis that is static, modular, sound, and precise. Our hope is that this analysis can serve as a building block for future static analyses. To bootstrap this future, we have used \sysname{} to implement prototypes of a program slicer and an IFC checker, shown in \Cref{fig:apps}. 

The program slicer in \Cref{fig:slicer} is a VSCode extension that fades out all lines of code that are not relevant to the user's selection, i.e. not part of the modular slice. Rather than present a slice of the entire program like in prior slicing tools, we can use Flowistry's modular analysis to present lightweight slices of just within a given function. Users can apply the slicer for comprehension tasks such as reducing the scope of a bug, or for refactoring tasks such as removing an aspect of a program like timing or logging. 

The IFC checker in \Cref{fig:ifc} is a Rust library and compiler plugin. It provides the user a library with the traits \verb|Secure| and \verb|Insecure| to indicate the relative security of data types and operations. Then the compiler plugin uses \sysname{} to determine whether information flows from \verb|Insecure| variables to \verb|Secure| variables. Users can apply the IFC checker to catch sensitive data leaks in an application. This prototype is purely intraprocedural, but future work could build an interprocedural analysis by using Flowistry's output as procedure summaries in a larger information flow graph.

\section{Related Work}\label{sec:rw}

Our work draws on three core concepts: information flow, modular static analysis, and ownership types.

\paragraph{Information flow}

Information flow has been historically studied in the context of security, such as ensuring low-security users of a program cannot infer anything about its high-security internals. Security-focused information flow analyses have been developed for Java\,\cite{myers1999jflow}, Javascript\,\cite{austin2009efficient}, OCaml\,\cite{pottier2003information}, Haskell\,\cite{stefan2011flexible}, and many other languages. 

Each analysis satisfies some, but not all, of our requirements from \Cref{sec:intro}. For instance, the JFlow\,\cite{myers1999jflow} and Flow-Caml\,\cite{pottier2003information} languages required adding features to the base language, violating our second requirement. Some methods like that of \citet{austin2009efficient} for Javascript rely on dynamic analysis, violating our third requirement. And Haskell only supports effects like mutation through monads, violating our first requirement.

Nonetheless, we draw significant inspiration from mechanisms in prior work. Our analysis resembles the slicing calculus of \citet{abadi1999core}. The use of lifetimes for modular analysis of functions is comparable to security annotations in Flow-Caml\,\cite{pottier2003information}. The CFG analysis draws on techniques used in program slicers, such as the LLVM dg slicer\,\cite{llvmslicer}. 

\paragraph{Modular static analysis}

The key technique to making static analysis modular (or ``compositional'' or ``separate'') is symbolically summarizing a function's behavior, so that the summary can be used without the function's implementation. Starting from \citet{rountev1999data} and \citet{cousot2002modular}, one approach has been to design a system of ``procedure summaries'' understood by the static analyzer and distinct from the language being analyzed. This approach has been widely applied for static analysis of null pointer dereferences\,\cite{yorsh2008generating}, pointer aliases\,\cite{dillig2011precise}, data dependencies\,\cite{tang2015summary}, and other properties.

Another approach, like ours, is to leverage the language's type system to summarize behavior. \citet{tang1994separate} showed that an effect system could be used for a modular control-flow analysis. Later work in Haskell used its powerful type system and monadic effects to embed many forms of information flow control into the language\,\cite{li2006encoding, russo2008library,stefan2011flexible,buiras2015hlio}.

\paragraph{Ownership types}

Rust and Oxide's conceptions of ownership derive from \citet{clarke1998ownership} and \citet{grossman2002region}. For instance, the Cyclone language of Grossman et al. uses regions to restrict where a pointer can point-to, and uses region variables to express relationships between regions in a function's input and output types. A lifetime is similar in that it annotates the types of pointers, but differs in how it is analyzed.

Recent works have demonstrated innovative applications of Rust's type system for modular program analysis. \citet{astrauskas2019leveraging} embed Rust programs into a separation logic to verify pre/post conditions about functions. \citet{jung2020stacked} use Rust's ownership-based guarantees to implement more aggressive program optimizations.

Closer to our domain, \citet{balasubramanian2017system} implemented a prototype IFC system for Rust by lowering programs to LLVM and verifying them with SMACK\,\cite{rakamaric2014smack}, although their system is hard to contrast with ours given the high-level description in their paper. \citet{njor2021static} implemented a static taint analysis for Rust, although it is not field-sensitive, alias-sensitive, or modular.

\section{Discussion}
\label{sec:discussion}

Looking forward, two interesting avenues for future work on \sysname{} are improving its precision and addressing soundness in unsafe code. For instance, the lifetime-based pointer analysis is sound but imprecise in some respects. Lifetimes often lose information about part-whole relationships. Consider the function that returns a mutable pointer to a specific index in a vector:
\begin{lstlisting}
fn get_mut<'a>(&'a mut self, i: usize) 
  -> Option<&'a mut T>;
\end{lstlisting}
These lifetimes indicate only that the return value points to \textit{something} in the input vector. The expressions \verb|v.get_mut(i)| and \verb|v.get_mut(i + 1)| are considered aliases even though they are not. Future work could integrate \sysname{} with verification tools like Prusti\,\cite{astrauskas2019leveraging} to use abstract interpretation for a more precise pointer analysis in such cases.

Additionally, Rust has many libraries built on unsafe code that can lose annotations essential to information flow, such as interior mutability. For example, shared-memory concurrency in Rust looks like this:
\begin{lstlisting}
let n: Arc<Mutex<i32>> = Arc::new(Mutex::new(0));
let n2: Arc<Mutex<i32>> = Arc::clone(&n);
*n2.lock().unwrap() = 1;
\end{lstlisting}
\verb|Arc::clone| does not share a lifetime between its input and output, so a lifetime-based pointer analysis therefore cannot deduce that \verb|n2| is an alias of \verb|n|, and \sysname{} would not recognize that the mutation on line 3 affects \verb|n|. Future work can explore how unsafe libraries could be annotated with the necessary metadata needed to analyze information flow, similar to how RustBelt\,\cite{jung2017rustbelt} identifies the pre/post-conditions needed to ensure type safety within unsafe code.

Overall, we are excited by the possibilities created by having a practical information flow analysis that can run today on any Rust program. Many exciting systems for tasks like debugging\,\cite{ko2004designing}, example generation\,\cite{head2018interactive}, and program repair\,\cite{wen2018context} rely on information flow in some form, and we hope that \sysname{} can support the development of these tools.

\begin{acks}
This work was partially supported by the Italian Ministry of Education through funding for the Rita Levi Montalcini grant (call of 2019).
\end{acks}

\appendix
\section{Appendix}
\label{sec:appendix}

This appendix contains three major sections. First, in \Cref{sec:more_rules} we include information flow inference rules beyond those presented in \Cref{sec:analysis} for the remaining features of Oxide. Second, in \Cref{sec:appendix_proofs} we provide a complete proof of the noninterference theorem in \Cref{sec:soundness}, along with associated lemmas. Finally, in \Cref{sec:eval_details} we provide additional details necessary to replicate the evaluation.

\subsection{Additional rules}
\label{sec:more_rules}

\begin{figure*}
\begin{mathpar}
\typerule{T-Tuple}
  {\forall i ~ . ~ 
    \tcnew{\Sigma}{\Delta}{\Gamma_{i-1}}{\hlight{\Theta_{i-1}}}
     {\hat{e}_i}{\tyvar{si}_i \hlight{\withslice{\kappa_i}}}{\Gamma_i}{\hlight{\Theta_i}}
    \\\\
    \hlight{\kappa = \{\loc\} \cup \medcup_i \kappa_i}
  }
  {\tcnew{\Sigma}{\Delta}{\Gamma_0}{\hlight{\Theta_0}}
   {(\hat{e}_1, \ldots, \hat{e}_n)_{\hlight{\loc}}}{(\tyvar{si}_1, \ldots, \tyvar{si}_n) \hlight{\withslice{\kappa}}}
   {\Gamma_n}{\hlight{\Theta_n}}}
  {ttuple}

\typerule{T-Seq}
  {
    \tcnew{\Sigma}{\Delta}{\Gamma}{\hlight{\Theta}}{e_1}{\tyvar{si}_1 \hlight{\withslice{\_}}}{\Gamma_1}{\hlight{\Theta_1}} 
    \\\\
    \tcnew{\Sigma}{\Delta}{\text{gc-loans}(\Gamma_1)}{\hlight{\Theta_1}}{e_2}{\tyvar{si}_2 \hlight{\withslice{\kappa_2}}}{\Gamma_2}{\hlight{\Theta_2}}
  }
  {
    \tcnew{\Sigma}{\Delta}{\Gamma}{\hlight{\Theta}}{e_1; e_2}{\tyvar{si}_2 \hlight{\withslice{\kappa_2}}}{\Gamma}{\hlight{\Theta_2}}
  }{tseq}

\typerule{T-LetProv}
  {\tcnew{\Sigma}{\Delta}{\Gamma, r \mapsto \varnothing}{\hlight{\Theta}}
   {e}{\tyvar{si} \hlight{\withslice{\kappa}}}{\Gamma', r \mapsto \loanset}{\hlight{\Theta'}} 
   \\
  }
  {\tcnew
    {\Sigma}{\Delta}{\Gamma}{\hlight{\Theta}}
    {\letprov{r}{e}}
    {\tyvar{si} \hlight{\withslice{\kappa}}}
    {\Gamma'}
    {\hlight{\Theta'}}}
  {tletprov}  
  
\typerule{T-Borrow}
  {\gray{\Gamma(r) = \varnothing}
   \\
   \Delta; \Gamma \vdash_\omega p \Rightarrow \loanset 
   \\\\
   \gray{\Delta; \Gamma \vdash_\omega p : \tyvar{xi}}
   \\
   \hlight{\kappa = \{\loc\} \cup \medcup_{^\omega p' \in \loanset} \Theta(p')}
   }
  {\tcnew
    {\Sigma}{\Delta}{\Gamma}{\hlight{\Theta}}
    {(\eref{\omega}{r}{p})_{\hlight{\loc}}}
    {\eref{\omega}{r}{\tyvar{xi}} \hlight{\withslice{\kappa}}}
    {\Gamma[r \mapsto \loanset]}
    {\hlight{\Theta}}}
  {tborrow}
  
\typerule{T-Copy}
  {
    \Delta; \Gamma \vdash_\shrd p \Rightarrow \loanset 
    \\
    \gray{\Delta; \Gamma \vdash_\shrd p : \tyvar{si}}
    \\\\
    \gray{\msf{copyable}_\Sigma \tyvar{si}}
    \\
    \hlight{\kappa = \medcup_{^\omega p' \in \loanset} \Theta(p')}
  }
  {
    \tcnew{\Sigma}{\Delta}{\Gamma}{\hlight{\Theta}}{p}{\tyvar{si} \hlight{\withslice{\kappa}}}{\Gamma}{\hlight{\Theta}}
  }{tcopy}
   
\typerule{T-Branch}
  {
    \tcnew{\Sigma}{\Delta}{\Gamma}{\hlight{\Theta}}{e_1}{\msf{bool} \hlight{\withslice{\kappa_1}}}{\Gamma_1}{\hlight{\Theta_1}} 
    \\
    \tcnew{\Sigma}{\Delta}{\Gamma_1}{\hlight{\Theta_1}}{e_2}{\tyvar{si}_2 \hlight{\withslice{\kappa_2}}}{\Gamma_2}{\hlight{\Theta_2}} 
    \\\\
    \tcnew{\Sigma}{\Delta}{\Gamma_1}{\hlight{\Theta_1}}{e_3}{\tyvar{si}_3 \hlight{\withslice{\kappa_3}}}{\Gamma_3}{\hlight{\Theta_3}} 
    \\
    \gray{\tyvar{si} = \tyvar{si}_2 \vee \tyvar{si} = \tyvar{si}_3}
    \\\\
    \gray{\Delta; \Gamma_2 \vdash \tyvar{si}_2 \lesssim \tyvar{si} \Rightarrow \Gamma_2'}
    \\
    \gray{\Delta; \Gamma_3 \vdash \tyvar{si}_3 \lesssim \tyvar{si} \Rightarrow \Gamma_3'}
    \\\\
    \gray{\Gamma_2' \doublecup \Gamma_3' = \Gamma' }
    \\
    \hlight{\Theta_2 \doublecup \Theta_3 = \Theta'} 
    \\ 
    \hlight{\Theta'' = \Theta'[\forall p \in \Theta' \setminus \Theta_1 ~ . ~ p \mapsto \Theta'(p) \cup \kappa_1]} 
  }
  {
    \tcnew{\Sigma}{\Delta}{\Gamma}{\hlight{\Theta}}
    {\msf{if}~e_1~\{~e_2~\}~\msf{else}~\{~e_3~\}}
    {\tyvar{si} \hlight{\withslice{\kappa_1 \cup \kappa_2 \cup \kappa_3}}}
    {\Gamma'}{\hlight{\Theta''}}
  }{tbranch}
\end{mathpar}
\caption{Additional information flow inference rules.}
\label{fig:additional_rules}
\end{figure*}

In \Cref{fig:additional_rules}, we provide the remaining information flow rules for the expression forms not covered in \Cref{sec:analysis}. \Cref{tr:ttuple}, \Cref{tr:tseq}, \Cref{tr:tletprov} are relatively straightforward, so we focus on the remaining three rules.

First, \Cref{tr:tborrow}: a borrow is unique in that its runtime value is a \textit{place}, i.e. $\msf{ptr}~\pi$. A borrow $\&\pi$ will always have value $\msf{ptr}~\pi$, so it has no dependencies. But the borrow is a reborrow, i.e. of the form $\&{\ast p}$, then we need to include the dependencies of all places that $p$ could point-to. Therefore we include $\Theta(p')$ for each loan $p'$ in the loan set of $p$.

Next, \Cref{tr:tcopy} shows a read from an arbitrary place expression $p$. Unlike \Cref{tr:tmove}, we have to account for $p$ referring to many possible memory locations. Again this is captured by the ownership-safety judgment $\vdash_\shrd p \Rightarrow \loanset$. Therefore the dependencies of $p$ are the dependencies of any possibly read place, i.e. $\medcup_{\pi'} \Theta(\pi')$.

Finally, \Cref{tr:tbranch} shows how to handle conditional execution. The $\kappa$ is simple, as the value of an if-expression could depend on either branch, $\kappa_2 \cup \kappa_3$, along with the condition, $\kappa_1$. The more complex aspect is handling effects in $\Theta$. The core idea is that if a place $p$ could be mutated in either $e_2$ or $e_3$, then that place has a control dependency on $e_1$, and $\kappa_1$ should be part of the dependencies of $p$.

To encode this idea, we introduce two new metafunctions. First, $\Theta_1 \doublecup \Theta_3$ represents the union distributed over like entries:
\[
    \Theta_2 \doublecup \Theta_3 \eqdef \{p \mapsto \Theta_2(p) \cup \Theta_3(p) \mid p \in \Theta_2 \vee p \in \Theta_3\}
\]
Next, we represent ``could be mutated in $e_2$ or $e_3$'' via $\Theta' \setminus \Theta_1$, similarly distributed over like entries:
\begin{align*}
    \Theta' \setminus \Theta_1 \eqdef \{&p \mapsto \Theta'(p) \setminus \Theta_1(p) \\ 
    &\mid (p \in \Theta' \vee p \in \Theta_1) \wedge \Theta'(p) \setminus \Theta_1(p) \neq \varnothing\}
\end{align*}
Then the rule says: after independently computing the contexts $\Theta_2$ and $\Theta_3$, compute a unioned context $\Theta'$. Then for all places $p$ that had new dependencies generated in $e_2$ or $e_3$, i.e. $p \in \Theta' \setminus \Theta_1$, add $\kappa_1$ to the dependencies of $p$.

\subsection{Proofs}
\label{sec:appendix_proofs}

The proof of non-interference relies on a few key lemmas about the semantics of Oxide. We start by defining and proving these lemmas (\Cref{sec:appendix_lemmas}), and then proceed to prove noninterference (\Cref{sec:appendix_noninterference}).

\subsubsection{Lemmas}
\label{sec:appendix_lemmas}

\begin{lemma}[Mutating a place also mutates its conflicts]
\label{lem:places}
Let:
\begin{itemize}[leftmargin=*]
    \item $\pi_\mut = \pi_\mut^\square[x]$, $\sigma$ where $\sigma \vdash \pi_\mut^\square \times x \Downarrow \valuectx$
    \item $v$ be a value and $\stepped{\sigma} = \sigma[x \mapsto \valuectx[v]]$
    \item $\pi_\any \in \sigma$
\end{itemize}
Then $\sigma(\pi_\any) \neq \stepped{\sigma}(\pi_\any) \implies \notdisjoint{\pi_\mut}{\pi_\any}$.
\end{lemma}

\begin{proof} \ 
\begin{proofsteps}
    \item Assume $\sigma(\pi_\any) \neq \stepped{\sigma}(\pi_\any)$. Want to show $\notdisjoint{\pi_\mut}{\pi_\any}$.
    \item Because only $x$ is assigned, then $\pi_\any = \pi_\any^\square[x]$. 
    \item Assume for sake of contradiction that $\disjoint{\pi_\mut}{\pi_\any}$. Let $q$ be the shared path and $n_\mut, n_\any$ be the split, i.e. $\pi_\mut = x.q.n_\mut.q_\mut$ and $\pi_\any = x.q.n_\any.q_\any$.
    \item By \textsc{ER-Projection}, then $$\sigma \vdash x.q.n_\mut \Downarrow \valuectx[(v_0, \ldots, v_{n_\any}, \ldots, \square_{n_\mut}, \ldots, v_n)]$$
    \item $v_{n_\any} = \sigma(x.q.n_\any)$ by induction on the derivation of $\valuectx$.
    \item Therefore $\stepped{\sigma}(x.q.n_\mut) = v_{n_\any} = \sigma(x.q.n_\mut)$. 
    \item This is a contradiction with $\sigma(\pi_\any) \neq \stepped{\sigma}(\pi_\any)$, therefore $\notdisjoint{\pi_\mut}{\pi_\any}$.
\end{proofsteps}
\end{proof}

\newcommand{\nest}{-1pt}

\begin{lemma}[A place expression's loan set contains the place it points-to at runtime.]
\label{lem:pointers}
Let:
\begin{itemize}[leftmargin=*]
    \item $\sigma, \Sigma, \Gamma$ where $\Sigma \vdash \sigma : \Gamma$
    \item $p_\mut$ where $\bullet; \Gamma \vdash_\uniq{} p_\mut \Rightarrow \loanset$ and $\sigma \vdash p_\mut \Downarrow \pi_\mut$
    \item $p_\any$ where $\sigma \vdash p_\any \Downarrow \pi_\any$
\end{itemize}
Then $\notdisjoint{\pi_\any}{\pi_\mut} \implies \exists \,^\uniq{} p_\loan \in \loanset ~ . ~ \notdisjoint{p_\any}{p_\loan}$
\end{lemma}

\begin{proof}
Proof by induction on the derivation of \\ $\bullet; \Gamma \vdash_\uniq{} p_\mut \Rightarrow \loanset$
\begin{proofcases}
    \item \textsc{O-SafePlace}:
    \begin{proofsteps}[\nest]
        \item If $p_\any \neq \pi_\any$ then $\exists r ~ . ~ \pi_\any \in \Gamma(r)$. But the first premise of \textsc{O-SafePlace}, it must be the case that $\disjoint{\pi_\any}{\pi_\mut}$, a contradiction. Therefore $p_\any = \pi_\any$.
        \item Then by the conclusion of \textsc{O-SafePlace}, \\ $\loanset = \{\,^\uniq \pi_\mut\,\}$. Then the theorem holds for $p_\loan = \pi_\mut$.
    \end{proofsteps}
    \item \textsc{O-Deref}:
    \begin{proofsteps}[\nest]
        \item Let $p_\mut = p_\mut^\square[\ast \pi_\mutptr]$ and $\Gamma(\pi_\mutptr) = \eref{\uniq}{r}{\tau}$. 
        \item By extension of Lemma E.6\,\citeyearpar[p.\,47]{weiss2019oxide}, 
        $$\exists \,^\uniq p_i \in \Gamma(r) ~ . ~ \sigma \vdash p_\mut^\square[p_i] \Downarrow \pi_\mut$$
        \item By the inductive hypothesis on 
        $$\bullet; \Gamma \vdash_\uniq p_\mut^\square[p_i] \Rightarrow \{\overline{^\uniq p'_i}\}$$ then: $$\exists \,^\uniq p_\loan \in \{\overline{^\uniq p'_i}\} ~ . ~ \notdisjoint{p_\any}{p_\loan}$$
        \item Because $\{\overline{^\uniq p'_i}\} \subseteq \loanset$, then the theorem holds.
    \end{proofsteps}
    \item \textsc{O-DerefAbs}: does not apply since $\Delta = \bullet$.
\end{proofcases}
\end{proof}

\begin{lemma}[A function only mutates unique references in its argument]
\label{lem:mutargs}
Let:
\begin{itemize}[leftmargin=*]
\item $\Gamma, \pi_\arrg, \sigma$ where $\Gamma(\pi_\arrg) = \tyvar{si}$ and $\Sigma \vdash \sigma : \Gamma$
\item $f$ where $\evalsto{\sigma}{f(\pi_\arrg)}{\stepped{\sigma}}{\_}$
\item $\stepped{\sigma}' =  \stepped{\sigma}[\forall p_\loan \in \loans{\uniq}{\pi_\arrg, \tyvar{si}, \bullet, \Gamma} ~ . ~ p_\loan \mapsto \sigma(p_\loan)]$
\end{itemize}
Then $\sigma = \stepped{\sigma}'$.
\end{lemma}

\begin{proof}\
\begin{proofsteps}
    \item Let: $$\Sigma(f) = \fndef$$
    \item By \textsc{E-App}, \textsc{E-EvalCtx}, and \textsc{E-Framed}: 
    \begin{align*}
    \Sigma \vdash &(\sigma; f(\pi)) \rightarrow (\sigma \stacksep x \mapsto \sigma(\pi);~\framed{e}) \xrightarrow{\ast} \\ &(\stepped{\sigma} \stacksep \varsigma;~ \framed{v}) \rightarrow (\stepped{\sigma};~ v)
    \end{align*}
    \item By inspection of the operational semantics, the only rule that could modify $\sigma$ (as apart from $\varsigma$) is \textsc{E-Assign}. Assume that $p_\mut := e'$ is executed under stack $\sigma_\mut \stacksep \varsigma$ during $f(\pi_\arrg)$ where $\sigma_\mut \vdash p_\mut \Downarrow \pi_\mut$.  
    \item By inspection of the operational semantics, the only way to create a pointer is \textsc{E-Borrow} on a place $p$. The only places in $\sigma_\mut$ that are accessible from $e$ are \\ $\loans{\shrd}{\pi_\arrg, \tyvar{si}, \bullet, \Gamma}$, so it must be that $p_\mut = p^\square[p_\loan]$ and $p_\loan \in \loans{\shrd}{\pi_\arrg, \tyvar{si}, \bullet, \Gamma}$.
    \item Because $e$ is well-typed under $x : \tyvar{si}_a$, then any shared references in $x$ cannot be mutated: \textsc{T-AssignDeref} requires $\vdash_\uniq p_\mut \Rightarrow \loanset$, which by \textsc{O-Deref} requires that the reference under $p_\mut$ has type $\eref{\uniq}{r}{\tau}$. Therefore $p_\loan \in \loans{\uniq}{\pi_\arrg, \tyvar{si}, \bullet, \Gamma}$.
    \item Hence, all mutated places in $\stepped{\sigma}$ are projections of a $p_\loan$. Therefore $\stepped{\sigma}[\forall p_\loan \in \loans{\uniq{}}{\pi_\arrg, \tyvar{si}, \bullet, \Gamma} ~ . ~ p_\loan \mapsto \sigma(p_\loan)]$ reverts all possible mutations, and $\sigma = \stepped{\sigma}'$.
\end{proofsteps}
\end{proof}

\begin{lemma}[A function's effects are only influenced by its argument.]
\label{lem:eqarg}
Let:
\begin{itemize}
\item $\Gamma, \pi, \sigma_i$ where $\Gamma(\pi_\arrg) = \tyvar{si}$ and $i \in \{1, 2\}$ and $\Sigma \vdash \sigma_i : \Gamma$
\item $f$ where $\evalsto{\sigma_i}{f(\pi_\arrg)}{\stepped{\sigma}_i}{v_i}$
\item $P = \{\pi_\arrg\} \cup \loans{\shrd}{\pi_\arrg, \tyvar{si}, \bullet, \Gamma}$
\end{itemize}

$\sigma_1 \stackeq{P} \sigma_2 \implies \stepped{\sigma}_1 \stackeq{P} \stepped{\sigma}_2 \wedge v_1 = v_2$
\end{lemma}

\begin{proof}\
\begin{proofsteps}
    \item As with the proof for \Cref{lem:mutargs}, have $(\sigma_i, f(\pi_\arrg)) \xrightarrow{\ast} (\stepped{\sigma}_i; v_i)$.
    \item By inspection of the operational semantics, the only rule that could read $\sigma$ is \textsc{E-Copy} via the premise $(\sigma_\rd; p_\rd) \rightarrow (\sigma_\rd; v)$ where $\sigma_\rd \vdash p_\rd \Downarrow \pi_\rd \mapsto v$.
    \item As with the proof for \Cref{lem:mutargs}, the only possible value of $p_\rd = p^\square[p_\loan]$ where $p_\loan \in \loans{\shrd}{\pi_\arrg, \tyvar{si}, \bullet, \Gamma}$.
    \item Because $p_\loan \in P$ and $\sigma_1 \stackeq{P} \sigma_2$ then $\sigma_1(p_\loan) = \sigma_2(p_\loan)$.
    \item Similarly because $\pi_\arrg \in P$ then $\sigma_1(\pi_\arrg) = \sigma_2(\pi_\arrg)$.
    \item Therefore every readable input for $f$ is equal, and for the same function body $e$, the evaluation should result in the same effects on $\sigma$ and same output $v$. Hence, $\stepped{\sigma}_1 \stackeq{P} \stepped{\sigma}_2$ and $v_1 = v_2$. 
\end{proofsteps}
\end{proof}

\subsubsection{Noninterference}
\label{sec:appendix_noninterference}

\begin{theorem}[Noninterference]
Let $e$ such that 
$$\tcnew{\Sigma}{\bullet}{\Gamma}{\Theta}{e}{\tau \withslice{\kappa}}{\Gamma'}{\Theta'}$$
For $i \in \{1, 2\}$, let $\sigma_i$ such that 
$$\Sigma \vdash \sigma_i : \Gamma \hspace{12pt} \text{and} \hspace{12pt} \evalsto{\sigma_i}{e}{\stepped{\sigma}_i}{v_i}$$
Then:
\begin{enumerate}[(a)]
    \item $\sigma_1 \stackdepeq{\Theta}{\kappa} \sigma_2 \implies v_1 = v_2$
    \item $\forall p \mapsto \kappa_p \in \Theta' ~ . ~ \sigma_1 \stackdepeq{\Theta}{\kappa_p} \sigma_2 \implies \stepped{\sigma}_1(p) = \stepped{\sigma}_2(p)$
\end{enumerate}
\end{theorem}

\begin{proof}
Proof by induction over the derivation of $e : \tau$. 
\begin{itemize}[itemsep=8pt,leftmargin=*]
    \item \Cref{tr:tu32}: $e = n_\loc$ where $\kappa = \{\loc\}$. (a) is trivial because $n$ always evaluates to $n$, and (b) is trivial because $n$ has no effects.
    
     \item \Cref{tr:tmove}: $e = \pi$ where $\kappa = \Theta(\pi)$. (b) is trivial because $\pi$ has no effects, so we focus on (a): $\sigma_1 \stackdepeq{\Theta}{\Theta(\pi)} \sigma_2 \implies v_1 = v_2$
    \begin{proofsteps}
        \item By the definition of equivalence of stacks, $\Theta(\pi) \subseteq \Theta(\pi) \implies \sigma_1(\pi) = \sigma_2(\pi)$.
        \item By \textsc{E-Move}, $v_i = \sigma_i(\pi)$. 
        \item Therefore $v_1 = v_2$.
    \end{proofsteps}
    
    \item \Cref{tr:tcopy}: the proof is the same as for \Cref{tr:tmove}.
    
    \item \Cref{tr:tseq}: $e = e_1; e_2$ where $\kappa = \kappa_2$.
    \begin{proofcases}
        \item $\sigma_1 \stackdepeq{\Theta}{\kappa_2} \sigma_2 \implies v_1 = v_2$
        \begin{proofsteps}[\nest]
            \item By \textsc{E-EvalCtx}, $\evalsto{\sigma_i}{e_1}{\sigma^{e_1}_i}{ v^{e_1}_i}$.
            \item By \textsc{E-Seq}, $\stepsto{\sigma_i^{e_1}}{(v_i^{e_1}; e_2)} {\sigma_i^{e_1}}{e_2} \xrightarrow{\ast} (\sigma^{e_2}_i; v^{e_2}_i)$. WTS $v^{e_2}_1 = v^{e_2}_2$.
            \item By the IH for $e_2$, $v^{e_2}_1 = v^{e_2}_2$ if $\sigma^{e_1}_1
            \stackdepeq{\Theta_1}{\kappa_2} \sigma^{e_1}_2$.
            \item Let $\pi \mapsto \kappa_\pi \in \Theta_1$. WTS $\kappa_\pi \subseteq \kappa_2 \implies \sigma^{e_1}_1(\pi) = \sigma^{e_1}_2(\pi)$.
            \item By the IH for $e_1$, $\sigma^{e_1}_1(\pi) = \sigma^{e_1}_2(\pi)$ if $\sigma_1 \stackdepeq{\Theta}{\kappa_\pi} \sigma_2$.
            \item Let $\pi' \mapsto \kappa_{\pi'} \in \Theta_1$ such that $\kappa_{\pi'} \subseteq \kappa_\pi$. WTS $\sigma_1(\pi') = \sigma_2(\pi')$.
            \item By assumption, because $\kappa_\pi \subseteq \kappa_2 \wedge \kappa_{\pi'} \subseteq \kappa_\pi$, then $\kappa_{\pi'} \subseteq \kappa_2$. 
            \item By assumption, $\sigma_1(\pi') = \sigma_2(\pi')$.
        \end{proofsteps}
        
        \item $\forall \pi \mapsto \kappa_\pi \in \Theta_2 ~ . ~ \sigma_1 \stackdepeq{\Theta}{\kappa_\pi} \sigma_2 \implies \stepped{\sigma}_1(\pi) = \stepped{\sigma}_2(\pi)$ 
        \begin{proofsteps}[\nest]
            \item By the IH for $e_2$, $\stepped{\sigma}_1(\pi) = \stepped{\sigma}_2(\pi)$ if $\sigma^{e_1}_1 \stackdepeq{\Theta_1}{\kappa_\pi} \sigma^{e_1}_2$.
            \item The proof that $\sigma^{e_1}_1 \stackdepeq{\Theta_1}{\kappa_\pi} \sigma^{e_1}_2$ follows similarly as above.
        \end{proofsteps}
    \end{proofcases}
    
    \item \Cref{tr:tlet}: $e = ``\olet{x}{\tyvar{si}_a}{e_1}{e_2}"$. The proof of (b) follows from the proof for \cref{tr:tseq}. We focus on (a): $\sigma_1 \stackdepeq{\Theta}{\kappa_2} \sigma_2 \implies v_1 = v_2$
    \begin{proofsteps}
        \item By \textsc{E-EvalCtx}, \textsc{E-Let}, \textsc{E-EvalCtx}, \textsc{E-Shift}: \\ 
        $\evalsto{\sigma_i}{\msf{let}~x : \tyvar{si}_a = e_1; e_2}{\sigma^{e_1}_i}{\msf{let}~x : \tyvar{si}_a = v_{e_1}^i;~ e_2} \rightarrow (\sigma^{e_1}_i, x \mapsto v_{e_1}^i;~ \msf{shift}~e_2) \xrightarrow{\ast} \\ (\stepped{\sigma}_i, x \mapsto \_;~ \msf{shift}~v^{e_2}_i) \rightarrow (\stepped{\sigma}_i;~ v^{e_2}_i)$.
        \item Let $\sigma^x_i = \sigma^{e_1}_i, x \mapsto v^{e_1}_i$. By IH (a) for $e_2$, if $\sigma^x_1 \stackdepeq{\Theta_1'}{\kappa_2} \sigma^x_2$ then $v^{e_2}_1 = v^{e_2}_2$.
        \item Let $\pi$ such that $\Theta_1'(\pi) \subseteq \kappa_2$. WTS $\sigma^x_1(\pi) = \sigma^x_2(\pi)$.
        \item If $\Theta_1'(\pi) \subseteq \Theta(\pi)$, then $\sigma^x_1(\pi) = \sigma^x_2(\pi)$ is true by assumption.
        \item Otherwise must be $\pi = \pi^\square[x]$ and $\kappa_q \subseteq \Theta_1'(\pi)$. 
        \item By IH (a) for $e_1$, if $\sigma_1 \stackdepeq{\Theta}{\kappa_1} \sigma_2$ then $v^{e_1}_1 = v^{e_1}_2$.
        \item Because $\sigma_1 \stackdepeq{\Theta}{\kappa_2} \sigma_2$ and $\kappa_1 \subseteq \kappa_2$ by (3), then  $\sigma_1 \stackdepeq{\Theta}{\kappa_1} \sigma_2$ and $v^{e_1}_1 = v^{e_1}_2$.
        \item Therefore $\sigma^x_1(\pi) = \sigma^x_2(\pi)$.
    \end{proofsteps}
    
    \item \Cref{tr:tassign}: $e = ``\pi_\mut := e"$. Because $v_i = ()$ then (a) is trivial, so we focus on (b): \\ $\forall \pi_\any \mapsto \kappa_\any \in \Theta_1' ~ . ~ \sigma_1 \stackdepeq{\Theta}{\kappa_\any} \sigma_2 \implies \stepped{\sigma}_1(\pi_\any) = \stepped{\sigma}_2(\pi_\any)$
    \begin{proofsteps}
        \item By \textsc{E-EvalCtx}, $\evalsto{\sigma_i}{e}{\sigma^e_i}{v^e_i}$. By IH (b) on $e$, then $\sigma^e_1(\pi_\any) = \sigma^e_2(\pi_\any)$.
        \item By \textsc{E-Assign}, $\Sigma \vdash (\sigma^e_i;~\pi_\mut := v^e_i) \\ \rightarrow (\sigma^e_i[x \mapsto \valuectx_i[v^e_i]];~())$ where $\pi_\mut = \pi_\mut^\square[x]$ and $\sigma^e_i \vdash \pi_\mut \Downarrow \valuectx_i$.
        \item By \Cref{lem:places}, if $\sigma^e_i(\pi_\any) \neq \stepped{\sigma}_i(\pi_\any)$ then $\notdisjoint{\pi_\any}{\pi_\mut}$.
        \item By \Cref{tr:tassign}, because $\notdisjoint{\pi_\any}{\pi_\mut}$ then \\ $\Theta'_1(\pi_\any) = \Theta_1(\pi_\any) \cup \kappa_e$.
        \item Because $\kappa_e \subseteq \Theta'_1(\pi_\any)$, then $\sigma_1 \stackdepeq{\Theta}{\kappa_e} \sigma_2$.
        \item By IH (a) on $e$, then $v^e_1 = v^e_2$.
        \item Therefore $\stepped{\sigma}_1(\pi_\any) = \stepped{\sigma}_2(\pi_\any)$.
    \end{proofsteps}
    
    \item \Cref{tr:tassignderef}: $e = ``p_\mut := e"$. Like \textsc{T-Assign}, we focus on (b): $\forall \pi_\any \mapsto \kappa_\any \in \Theta_1' ~ . ~ \sigma_1 \stackdepeq{\Theta}{\kappa_\any} \sigma_2 \implies \stepped{\sigma}_1(\pi_\any) = \stepped{\sigma}_2(\pi_\any)$
    \begin{proofsteps}
        \item By \textsc{E-EvalCtx} and \textsc{E-Assign}, \\ $\evalsto{\sigma_i}{p_\mut := e}{\sigma^e_i[x_i \mapsto \valuectx_i[v^e_i]]}{()}$ \\ where $\sigma^e_i \vdash p_\mut \Downarrow \valuectx_i \times \pi_{\mut,i}^\square[x_i]$.
        \item By IH (b) on $e$, $\sigma^e_1(\pi_\any) = \sigma^e_2(\pi_\any)$.
        \item By \Cref{lem:places}, only consider the case where $\notdisjoint{\pi_{\mut,i}}{\pi_\any}$.
        \item By \Cref{lem:pointers}, because $\notdisjoint{\pi_{\mut,i}}{\pi_\any}$ then \\ $\exists \,^\uniq p_\loan \in \loanset ~ . ~ \notdisjoint{\pi_\any}{p_\loan}$.
        \item By \Cref{tr:tassignderef}, then $\kappa \subseteq \Theta_1'(\pi_\any)$, and $v^e_1 = v^e_2$, and $\stepped{\sigma}_1(\pi_\any) = \stepped{\sigma}_2(\pi_\any)$ as in the proof for \Cref{tr:tassign}.
    \end{proofsteps}
    
    \item \Cref{tr:tapp}: $e = ``f(\pi_\arrg)"$ where $\kappa = \kappa_\arrg$.
    \begin{proofcases}
        \item $\sigma_1 \stackdepeq{\Theta}{\kappa_\arrg} \sigma_2 \implies v_1 = v_2$
        \begin{proofsteps}[\nest]
            \item Let $p \in \{\pi_\arrg\} ~ \cup ~ \loans{\shrd}{\pi_\arrg, \tyvar{si}_a, \bullet, \Gamma_2}$. \\By \Cref{lem:eqarg}, if $\sigma_1(p) = \sigma_2(p)$, then $v_1 = v_2$. 
            \item By \Cref{tr:tapp}, $\Theta(p) \subseteq \kappa_\arrg$. Therefore $\sigma_1(p) = \sigma_2(p)$ and hence $v_1 = v_2$.
        \end{proofsteps}

        \item $\forall \pi_\any \mapsto \kappa_\any \in \Theta_2 ~ . ~ \sigma_1 \stackdepeq{\Theta}{\kappa_\any} \sigma_2 \implies \stepped{\sigma}_1(\pi_\any) = \stepped{\sigma}_2(\pi_\any)$
        \begin{proofsteps}[\nest]            
            \item By \Cref{lem:mutargs} and \Cref{lem:places}, if $\sigma_i(\pi_\any) \neq \stepped{\sigma}_i(\pi_\any)$ then: $\exists p_\loan \in \loans{\uniq}{\pi_\arrg, \tyvar{si}, \bullet, \Gamma_2} ~ . ~ (\sigma_i \vdash p_\loan \Downarrow \pi_{\loan,i}) \wedge \notdisjoint{\pi_\any}{\pi_{\loan,i}}$.
            \item By \Cref{tr:tapp}, then $\kappa_\arrg \subseteq \kappa_\any$. 
            \item By \Cref{lem:eqarg}, then $\stepped{\sigma}_1(\pi_\any) = \stepped{\sigma}_2(\pi_\any)$.
        \end{proofsteps}
    \end{proofcases}
    
\end{itemize}
\end{proof}

\subsection{Evaluation details}
\label{sec:eval_details}

\begin{table*}
\begin{tabular}{l | l | l | p{3cm}}
    \textbf{Project} & \textbf{Crate} & \textbf{Commit} & \textbf{Features} \\ \hline
    \href{https://github.com/rayon-rs/rayon/}{rayon} &  & \Verb|c571f8ffb4f74c8c09b4e1e6d9979b71b4414d07| & \emph{all} \\ \hline
    \href{https://github.com/SergioBenitez/Rocket}{Rocket} & core/lib & \Verb|8d4d01106e2e10b08100805d40bfa19a7357e900| & \emph{none} \\ \hline
    \href{https://github.com/ctz/rustls}{rustls} & rustls & \Verb|cdf1dada21a537e141d0c6dde9c5685bb43fbc0e| & \emph{all} \\ \hline
    \href{https://github.com/mozilla/sccache}{sccache} & & \Verb|3f318a8675e4c3de4f5e8ab2d086189f2ae5f5cf| & \emph{none} \\ \hline
    \href{https://github.com/dimforge/nalgebra}{nalgebra} & & \Verb|984bb1a63943aa68b6f26ff4a6acf8f68b833b70| & {\small rand, arbitrary, sparse, debug, io, libm} \\ \hline
    \href{https://github.com/image-rs/image}{image} & & \Verb|e916e9dda5f4253f6cc4557b0fe5fa3876ac18e5| & \emph{none} \\ \hline
    \href{https://github.com/hyperium/hyper}{hyper} & & \Verb|ed2fdb7b6a2963cea7577df05ddc41c56fee7246| & {\small full} \\ \hline
    \href{https://github.com/mrDIMAS/rg3d}{rg3d} & & \Verb|ca7b85f2b30e45b82caee0591ee1abf65bb3eb00| & \emph{all} \\ \hline
    \href{https://github.com/xiph/rav1e}{rav1e} & & \Verb|1b6643324752785e7cd6ad0b19257f3c3a9b2c6a| & \emph{none} \\ \hline
    \href{https://github.com/RustPython/RustPython}{RustPython} & vm & \Verb|9143e51b7524a5084d5ed230b1f2f5b0610ac58b| & {\small compiler}
\end{tabular}
\cprotect\caption{Build configuration for each package used in the evaluation. Crate is the sub-directory of the crate used for evaluation. Commit is the point in the project's Git history used. Features is the flags passed to \Verb|cargo|, where \emph{none} means ``no feature flags'' and \emph{all} means \Verb|--all-features|.}
\label{tab:build_config}
\end{table*}

\Cref{tab:build_config} provides the build configurations needed to precisely reproduce the codebases used in \Cref{sec:evaluation}. We also provide a Zenodo artifact containing a Docker image that can exactly reproduce the results in \Cref{sec:evaluation}: \url{https://zenodo.org/record/6327882}

\balance

\bibliography{bibliography}

\end{document}